\documentclass[journal,10pt,twoside]{IEEEtran}

\usepackage{fancyhdr}
\usepackage{hyperref}

\usepackage{amssymb}
\usepackage{amsmath}
\usepackage{cancel}
\usepackage{amsfonts}
\usepackage{subfig}
\usepackage{bbm}
\usepackage[numbers]{natbib}
\usepackage{graphicx}
\usepackage{epstopdf}
\usepackage{url}

\usepackage{lineno}
\usepackage{color}
\usepackage{amsthm}

\usepackage{caption}

\newtheorem{thm}{Theorem}[section]

\newtheorem{lemma1}{Theorem}[section]
\newtheorem{lemma}[lemma1]{Lemma}

\begin{document}

\title{Campaigning in Heterogeneous Social Networks: Optimal Control of SI Information Epidemics}

\author{Kundan Kandhway and Joy Kuri
\thanks{Authors are with the Department of Electronic Systems Engineering, Indian Institute of Science, Bangalore 560012, India; e-mail: \{kundan, kuri\} @dese.iisc.ernet.in}
}


\maketitle

\thispagestyle{fancy}

\begin{abstract}
We study the optimal control problem of maximizing the spread of an information epidemic on a social network. Information propagation is modeled as a Susceptible-Infected (SI) process and the campaign budget is fixed. Direct recruitment and word-of-mouth incentives are the two strategies to accelerate information spreading (controls). We allow for multiple controls depending on the degree of the nodes/individuals. The solution optimally allocates the scarce resource over the campaign duration and the degree class groups. We study the impact of the degree distribution of the network on the controls and present results for Erd\H os-R\'enyi and scale free networks. Results show that more resource is allocated to high degree nodes in the case of scale free networks but medium degree nodes in the case of Erd\H os-R\'enyi networks. We study the effects of various model parameters on the optimal strategy and quantify the improvement offered by the optimal strategy over the static and bang-bang control strategies. The effect of the time varying spreading rate on the controls is explored as the interest level of the population in the subject of the campaign may change over time. We show the existence of a solution to the formulated optimal control problem, which has non-linear isoperimetric constraints, using novel techniques that is general and can be used in other similar optimal control problems. This work may be of interest to political, social awareness, or crowdfunding campaigners and product marketing managers, and with some modifications may be used for mitigating biological epidemics.

\end{abstract}


\IEEEpeerreviewmaketitle

\section{Introduction}

\IEEEPARstart{P}{olitical} campaigners, crowdfunders and product marketing managers are using social networks increasingly to influence individuals. The size of the online social networks---in addition to the (old) human network where two individuals interacting in day-to-day life are connected via a link---is increasing day by day, which gives campaigners an opportunity to mold the opinion of many individuals. Information (\emph{e.g.} awareness of newly launched or upcoming products or services like smartphones, video games, satellite TV plans, movies etc; ideologies of political candidates) spreads through networks in a manner similar to pathogens on human networks, leading to ``information epidemics''. The campaigners' goal is to `infect' as many individuals as possible with the message by the campaign deadline, respecting her budget constraints. Money and manpower are scarce resources that need to be utilized judiciously over the duration of the campaign, to achieve optimal results. Such a resource allocation problem leads to the `optimal control' formulation---maximize an objective functional by adjusting a control which affects the system evolution.

We address the above resource allocation problem to maximize the fraction of nodes in a social network that are aware of the information. Individuals communicate the information to their neighbors (through Twitter tweets, Facebook posts, exchange of ideas when two individuals interact in face-to-face meetings etc.), giving rise to an information epidemic. The campaigner can influence this information spreading in two ways: (i) By directly recruiting individuals from the population at some cost by adjusting a \emph{direct recruitment control} (advertisements in mass media announcing discounts on products). (ii) By accelerating message spreading by incentivizing individuals who already have the message to spread more, by adjusting a \emph{word-of-mouth control, e.g.} announcing referral rewards in the form of discounts, coupons or cash-backs for introducing a friend to a product or service. Word-of-mouth incentives may be announced by emailing current customers, which will encourage them to put in a good word for the company. 

Resource limitations prevent the campaigner from communicating the information to the whole population. Not only is the timing of the direct recruitment and word-of-mouth incentives a crucial factor in determining the extent of information spreading, but also the resource distribution among different types of individuals and strategies. For example, individuals with large numbers of links are known to be influential spreaders \cite{goldenberg2009role}, and more resources may be spent on them.

\emph{Justification for using the Susceptible-Infected (SI) model:} We choose to model the information epidemic as the SI process (where a fixed fraction, $\alpha<1$, of infected individuals spread the information in the absence of any word-of-mouth control), over other possibilities such as the Susceptible-Infected-Susceptible (SIS) and Susceptible-Infected-Recovered (SIR) processes. The SI process has been used in previous studies as a preferred model for information propagation 
(\emph{e.g.} \cite{karsai2011small}). It is suitable to model situations where people who have ever received a message do not forget it during the campaign period and depending upon the level of word-of-mouth incentive, spread the message via their links. Such a situation may be encountered in marketing of a durable product (\emph{e.g.} latest version of smartphones or computer games) or services requiring memberships (\emph{e.g.} mobile-phone and satellite TV services) where the campaign duration is of the same order as the life of the product. In these cases, the contact details of the customers are always available with the company, which can incentivize their customers with varying levels of word-of-mouth control to turn them into active spreaders. Thus, the word-of-mouth incentive changes the fraction of active spreaders among the infected.

Infected individuals \emph{recover} in the SIR process or fall back to the susceptible state in the SIS process. SIS processes may be suitable to model marketing of consumables. SIR processes may be suitable for situations where individuals will stop spreading the message after random amounts of time and the campaign duration is large compared to the `recovery' time (\emph{e.g.} a situation arising in multi-level marketing, where an individual, although still a member, has become inactive).

The framework developed in this paper can be used in situations where only direct recruitment control can be applied, as well as those in which both direct recruitment and word-of-mouth controls can be applied. The former is suitable for modeling political campaigns, campaigns for spreading awareness of a social cause, movie promotion etc. The latter model is suitable for the product marketing as well as campaigns for crowdfunding. Crowdfunding refers to micro-funding in forms of loans, donations, equity purchases or pre-ordering of products yet to be produced, through the Internet and online social networks such as Facebook, Twitter, LinkedIn etc. \cite{belleflamme2013crowdfunding}. 

\subsection{Related Work and Our Contributions}

Devising strategies for epidemic prevention has been a subject of interest in many studies (\emph{e.g.} \cite{shaw2010enhanced, dykman2008disease}). 
Our work is different in that the \emph{theory of optimal control} is used. The literature on optimal control of disease and computer virus epidemics, in a homogeneously mixed population, is plentiful \citep{asano2008optimal, gaff2009optimal, yan2008optimal, zhu2012optimal}. 
\emph{Our work is different from the previous literature on optimal control of epidemics in two aspects:} \emph{First}, we consider information epidemics, which are \emph{required to be disseminated}; this seems to have attracted less attention compared to biological and computer virus epidemics, which need to be curbed. \emph{Secondly and more importantly}, we consider a population of \emph{networked individuals}. We note that networks play an important role in information dissemination, as individuals rarely interact randomly with others in society (as homogeneous mixing assumes); and even if they do, they seldom trust those unknown to them. Most people interact and trust a small group of individuals (`small' compared to total population size) who are their `neighbors' and share a link with them in the network.

Many of the papers above consider a homogeneously mixed population, and the ones which do not \cite{dayama2012optimal}, present results for the case where individuals are divided into a maximum of two degree classes only. In contrast, we consider almost three hundred degree classes in this study. Unlike the previous literature, this allows us to \emph{study the effect of network topology (Power Law, Erd\H os-R\'enyi) on control of information dissemination in networks.} Moreover in contrast to \cite{dayama2012optimal}, we formulate the optimal control problem with a fixed budget constraint. 

A recent study \cite{youssef2013mitigation} has formulated the optimal epidemic prevention problem for the network case, where mean field equations are written for each node of the network instead of each degree class, the approach taken in this study. However, \cite{youssef2013mitigation} computed the optimal solution only for a five node network and proposed heuristics for larger networks. 

The works on rumor spreading on technological and social networks with known topologies \citep{chierichetti2009rumor,pittel1987spreading} are similar in spirit to the problem considered here. However, the aim there was to compute maximum number of communication rounds required to disseminate a piece of information to almost all the nodes, following fixed communication strategies. Nodes may either `Pull' the message from their neighbors, `Push' the information to them or do both. Finding the optimal message spreading strategy was not considered in those works.

The authors in \cite{karnik2012optimal,kandhway2014run,kandhway2014optimal} and \cite{khouzani2011optimal} maximized information and security patch dissemination through optimal control on social and computer networks respectively, but considered homogeneously mixed population. Similarly, the author in \cite{belen2008behaviour} used impulse control to maximize the information spread in a homogeneously mixed population, assuming the Daley-Kendall and Maki-Thompson models, which are different from the Susceptible-Infected model assumed in this paper. The authors in \cite{sethi2008optimal} devised optimal pricing and advertisement strategies for new products without considering the epidemic message propagation, as is the case in this paper.

\emph{A unique aspect of an information epidemic} is that the interest level of people in talking about the subject of the campaign may undergo a gradual change over the campaign duration (\emph{e.g.,} decreasing interest for a model of a smartphone or a version of a software or computer game as they grow old; or increasing interest level of people to talk about the upcoming elections). Thus, the spreading rate of an information epidemic may vary during the campaign duration. This differentiates it from biological epidemics, where the spreading rate is constant. \emph{We have captured this phenomenon in this work, which differentiates it from the previous literature on disease and information epidemics.} For some applications like crowdfunding, the spreading rate is not expected to change over time.

Another contribution of this work is to \emph{show the existence of a solution to optimal control problems with a budget constraint under non-linear costs of applying controls.} Previous studies on biological epidemics on homogeneously mixed populations (\emph{e.g.} \cite{asano2008optimal,gaff2009optimal}) have considered non-linear cost functions (on controls); 
however they considered a linear combination of the cost and the reward reaped by the application of controls to the system, and minimized the 
\emph{net cost}. In contrast, we have explicit budget constraints. \cite{karnik2012optimal} does not have a network structure, although the authors have formulated a problem with an explicit budget constraint. \cite{dayama2012optimal} has a network structure, but does not have explicit budget constraints. 

Notice that the cost incurred by application of controls is often non-linear in controls in practical applications, as argued in \cite{gaff2009optimal}. 
This motivates us to formulate the problem with a non-linear cost structure.
Of course, this is also a generalization of the linear cost structure that has been assumed widely in earlier literature. Once we have an explicit budget with non-linear costs, the standard theorems (by Filippov and Cesari \cite{fleming1975deterministic}) for showing the existence of a solution are not applicable (because the system is now non-linearly influenced by the control) and hence new techniques for showing the existence of a solution are required.

Apart from considering a network structure and an explicit budget constraint for the first time in this paper, \emph{we have tailored the classical Susceptible-Infected model to capture information spreading.} In particular, unlike a biological epidemic, information spreading is voluntary. We have captured this by introducing a parameter $\alpha$ to represent the fraction of the infected population that chooses to be spreaders. Word-of-mouth control encourages slightly unwilling people to act as spreaders, increasing the effective $\alpha$. A similar control was considered in \cite{khouzani2011optimal} in a four state model for a homogeneously mixed population.

Our results show considerable improvement in the final fraction of the infected population (people who are aware of the topic), when compared to a static and  bang-bang control strategies, that respect the same budget constraints. The improvements, as expected, are very substantial compared to the no control (campaign) strategy.

\section{The System Model and Problem Formulation}
\label{sec:sys_model}

Individuals in the target population are organized into a static network (or a graph). The degree of a node (individual) is the number of neighbors the node in the network is connected to. The nodes are grouped into different classes. All the nodes with degree $k$ are said to be in the degree class $k$. The set of all possible degree classes (or degrees) is denoted by $\mathbb K$. Thus, $k\in \mathbb K=\{K_{min},...,K_{max}\}$, $K_{min}$ and $K_{max}$ denote the minimum and maximum degrees of the nodes in the network. There are a total of $|\mathbb{K}|$ degree classes where $|\mathbb{K}|$ denotes the cardinality of the set $\mathbb{K}$. Let the size of the network be $N$ and the number of nodes in the degree class $k\in \mathbb K$ be $N_k$, so that $\sum_{k\in \mathbb K}N_k = N$. Our formulation works for networks with an arbitrary \emph{degree distribution}, $p_k$. The degree distribution of a network is defined as the probability mass function, $p_k$ of the degrees $k$ of the nodes in the network, $\sum_{k\in \mathbb K}p_k = 1$. For a network, the empirical degree distribution is, $p_k=N_k/N,~k\in \mathbb K$ \cite{barrat2008dynamical}.

\emph{\textbf{Uncontrolled System:}} We model the Susceptible-Infected information epidemic by the `degree based compartmental model'. These models are most accurate on configuration model networks. These networks lack correlation in the way nodes are connected to one another, in the sense that a half edge of a given node is equally likely to be connected to any other half edge in the network. Also, being a mean field model, degree based compartmental model requires population size $N$ to be large.

The campaigner is interested in spreading a message in the population connected via the network. The process starts at time $t=0$ and the campaign deadline is $t=T$. The nodes in the network are classified into two categories---susceptible and infected. A susceptible node is yet to receive the message and an infected node already has it. Infected nodes are further classified into active (spreaders) and non-active nodes. Only $\alpha < 1$ fraction of the infected population is interested in transmitting the message further (active or spreaders). Thus, a newly infected node chooses to become a spreader with a probability $\alpha$. Effective fraction of spreaders in the infected population can be changed by the level of the word-of-mouth incentive (discussed later).

Let the numbers of susceptible and infected nodes at time $0\leq t\leq T$, in the degree class $k\in \mathbb K$ be $S_k(t)$ and $I_k(t)$ respectively; and, $s_k(t)=S_k(t)/N_k$, $i_k(t) = I_k(t)/N_k$. Then, the fractions of susceptible and infected nodes in the network at time $t$ is given by $s(t)=\sum_{k\in \mathbb K}p_ks_k(t)$ and $i(t)=\sum_{k\in \mathbb K}p_ki_k(t)$ respectively with $s(t)+i(t)=1$ and $s_k(t)+i_k(t)=1,~\forall~k\in \mathbb K$. 

The population consists mostly of the susceptibles at $t=0$, except for a small fraction $i_0$, of nodes which act as the seed for information dissemination. We assume $i_k(0)=i_0,~\forall k\in \mathbb K$. The message is passed probabilistically due to susceptible--spreader contact. At time $t$, a susceptible node acquires the message from a spreader at a rate $\beta(t),~t\in[0,T]$. In other words, in a small interval $dt$ at time $t$, a `susceptible' node changes its state to `infected' due to a single susceptible--spreader link with a probability $\beta(t) dt$.

The degree based compartmental model makes an approximation that all the nodes in the degree class $k$ behave in exactly the same way \citep[Sec. 17.10.2]{newman2010networks} \citep[Sec. 9.2]{barrat2008dynamical}. We first evaluate the probability that a tagged susceptible node of degree $k$ will change to infected state in a small interval $dt$ at time $t$. For the configuration model, the neighbor degree distribution, which is the probability that we will find a neighbor of degree $k$ if we follow an edge of any node in the network, is $r_k=kp_k/\bar k$ (details are in Appendix \ref{appendix:neighbor_degree_distribution}). Here $\bar k=\sum_{k\in \mathbb K}kp_k$ is average degree of the network. Note that $r_k$ is not simply $p_k$, one cannot reach a node with degree zero by following an edge even if $p_0\neq 0$. Since the node in question is susceptible and the neighbor is infected (and active), so the neighbor must have acquired the information from somewhere else. So the quantity of interest here is `excess degree distribution', which discounts the edge (of the neighbor) due to the tagged susceptible node we are situated at. Excess degree distribution is given by $q_k = r_{k+1} = (k+1)p_{k+1}/\bar k$. Note that, $q_{K_{max}}=0$.

The mean number of active neighbors around the tagged susceptible node of degree $k$, who can potentially pass the information to it, is $k\sum_{l\in \mathbb K}(\alpha i_l(t)q_l)$. Thus the tagged susceptible node will switch to infected state with a probability\footnote{It is 1 minus the probability that none of the neighbors pass the information. Information flow from the neighbors is assumed independent.} $1-(1-\beta(t)dt)^{k\sum_{l\in \mathbb K}(\alpha i_l(t)q_l)} \simeq \beta(t) k \sum_{l\in \mathbb K}(\alpha i_l(t)q_l) dt$. Since the fraction of susceptible nodes in degree class $k$ at time $t$ is $s_k(t)$, the total increase in the fraction of infected nodes in degree class $k$ in interval $dt$ is given by $s_k(t) \beta(t) k \sum_{l\in \mathbb K}(\alpha i_l(t)q_l) dt$ \citep[Sec. 17.10.2, adapted for time varying $\beta(t)$]{newman2010networks}. Thus, the rate of change of infected nodes in degree class $k$ in the uncontrolled SI epidemic is:
\begin{eqnarray}
\dot i_k(t) & = & \beta(t)k s_k(t) \sum_{l\in \mathbb K}(\alpha i_l(t)q_l);~~k\in \mathbb K. \nonumber
\end{eqnarray}

\emph{\textbf{Admissible Controls:}} To control the above system, we use the direct recruitment and word-of-mouth control signals which are defined next. We denote the $M$-dimensional direct recruitment vector control function by $\boldsymbol u=(u_1,...,u_M)$. The value $u_m(t),~1\leq m \leq M$, denotes the rate of direct recruitment, at time $t$, carried out by the campaigner in the group $m$ of degree classes, denoted by $\mathbb K_m=\{\tilde k_{m-1}+1,...,\tilde k_m\}$, $K_{min}-1=\tilde k_0<\tilde k_1 < ... < \tilde k_M = K_{max}$. Through direct recruitment control, the campaigner taps the pool of available susceptible nodes and converts them into infected ($\alpha$ fraction of those who get infected are spreaders). This can be done, for example, by placing advertisements in the mass media.

The purpose of having $M$ control functions and dividing degree classes into $M$ groups is to control each group with one control, thereby identifying groups which are important to target at any given time. Fig. \ref{fig:Km_explanation_1} shows a small network with $5$ nodes: $a$ to $e$, with degree sequence: $2,3,4,2,1$. Thus $\mathbb K = \{1,2,3,4\}$ and empirical degree distribution: $p_1=p_3=p_4=1/5,p_2=2/5$. If we decide to have $M=2$, we may choose $\mathbb K_1=\{1,2\}, \mathbb K_2=\{3,4\}$. Thus nodes $a,d,e$ will be controlled with $u_1$ and $b,c$ with $u_2$. Fig. \ref{fig:Km_explanation_2} shows an example of grouping a network with $|\mathbb K|=295$ degree classes into $M=3$ groups. We emphasize that even though degree classes are grouped into fewer groups for the purpose of controlling them, we will still capture information diffusion dynamics at a finer level with differential equations for each of the $|\mathbb K|$ classes.

\vspace{-1em}
\begin{figure}[ht!]
\centering
\subfloat[\label{fig:Km_explanation_1}]{
\includegraphics[width=40mm]{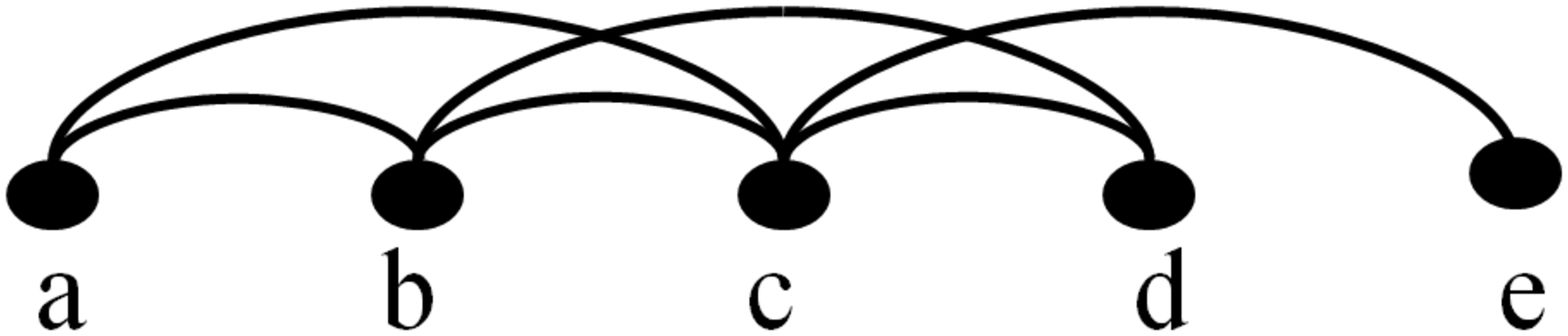} }
\hfill
\subfloat[\label{fig:Km_explanation_2}]{
\includegraphics[width=40mm]{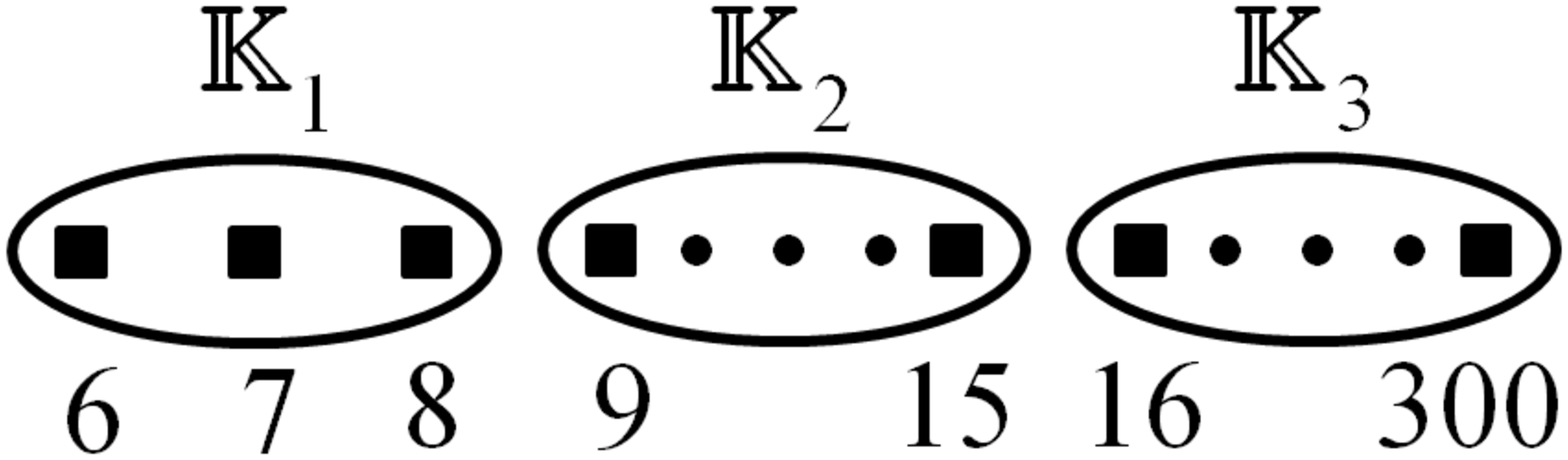} }
\caption{\small{(a) An example graph. Each filled oval represents a node. (b) An example of dividing a network with minimum degree $K_{min}=6$ and maximum degree $K_{max}=300$ into $M=3$ three degree class groups. Each square represents a degree class.}}
\label{fig:Km_explanation}
\end{figure}

The set of all admissible direct controls is defined in the following. Let $\Psi$ be the set of all equicontinuous functions over $[0,T]$, \emph{i.e.}, $|\sigma(t)-\sigma(\hat t)| \leq C_\Psi(\epsilon)$, for $t,\hat t \in [0,T]$; $|t - \hat t| \leq \epsilon$, $\forall \sigma \in \Psi$ with $C_\Psi(\epsilon)\rightarrow 0$ as $\epsilon \rightarrow 0$ \cite[Sec. 1.6]{atkinson2009theoretical}. Then,
\begin{align}
& u_m\in U_m \triangleq \{\sigma\in \Psi:0\leq \sigma(t)\leq u_{max},~\forall t\in[0,T] \}, \label{eq:U_m} \\
& \textnormal{and, } \boldsymbol u = (u_1,...,u_M) \in \boldsymbol U \triangleq \overset{M}{\underset{m=1}{\times}} U_m. \label{eq:U}
\end{align}
Practical considerations will require each of the $u_m(t)$ to be continuous and bounded (maximum allowed direct control is $u_{max}$) at all times $t$. We make a slightly stricter assumption that the functions in $\Psi$ are equicontinuous (\emph{i.e.} the same $C_\Psi(\epsilon)$ works for all the functions in $\Psi$ and it is independent of $t$) as it aids in showing existence of a solution to the optimal control problem (\ref{eq:opt_prob}) (will be discussed later). This assumption is not too strict and is milder than, for example, assuming differentiability of functions.

The word-of-mouth control affects the fraction of spreaders in the infected population and is denoted by $\boldsymbol v=(v_1,...,v_M)$. The value $v_m(t),~1\leq m \leq M$, denotes the rate at which word-of-mouth incentives are handed out, at time $t$, in the $m$th group of degree classes, $\mathbb K_m$, defined above. The set of all admissible word-of-mouth controls is defined as:
\begin{align}
& \boldsymbol v=(v_1,...,v_M) \in \boldsymbol V \triangleq \overset{M}{\underset{m=1}{\times}} V_m \label{eq:V} \\
& \textnormal{where, } v_m \in V_m \triangleq \{\sigma\in \Psi:0\leq \sigma(t)\leq v_{max},~\forall t\in[0,T]\}. \label{eq:V_m}
\end{align}
Here, $v_{max}$ is the maximum allowed word-of-mouth control signal and the set $\Psi$ was defined above.

\emph{\textbf{Controlled System:}} The objective (reward) functional for the optimal control problem is chosen to be, $J=\sum_{k\in \mathbb K} p_ki_k(T)$. For applications like political/social awareness/crowdfunding campaigns and product marketing, considered in this paper; we want to maximize the final number (fraction) of individuals who have received the message by the end of the campaign, \emph{i.e.} at $t=T$.  We do not care about the evolution history, $i_k(t),0<t<T$; this motivates such a choice for $J$. In addition, we assume that the resources at our disposal for applying the control signals are fixed. The budget constraint is captured in Eq. (\ref{eq:opt_prob_B_d}), where $B$ denotes the combined budget for both types of controls.

\emph{Rationale for the budget constraint (\ref{eq:opt_prob_B_d}):} The functions $b_m(.)$ and $c_m(.)$ capture the instantaneous cost of applying the direct recruitment and the word-of-mouth controls in the group $m$ and are continuous, non-negative and increasing functions in their arguments for $1\leq m\leq M$ (more effort incurs more cost). Also, $g_m=\sum_{k\in \mathbb K_m} p_k$ denotes the fraction of the whole population in group $m$. The instantaneous resource consumption for the direct control strategy is $\sum_{m=1}^Mg_mb_m(u_m(t))$. 

For the word-of-mouth control, the resource is spent when an active node demonstrates that it has successfully influenced a susceptible. Denote by $\bar s(t)$ the average fraction of susceptible nodes around a given node. Then $\bar s(t)=\sum_{k\in \mathbb K} r_ks_k(t)$, where $r_k=kp_k/\bar k$ is the neighbor degree distribution defined earlier. Given that we are in group $m$, the probability of picking up an active node of degree $k$ and a susceptible neighbor is $(p_k/g_m)\alpha i_k(t)\bar s(t)$. In a small interval $dt$, the susceptible neighbor will convert to infected with probability $(p_k/g_m)i_k(t)\bar s(t)\alpha v_m(t)\beta(t)dt$ due to application of word-of-mouth control. So the average word-of-mouth resource consumed for $k$ neighbors is $k c_m(v_m(t)) (p_k/g_m)i_k(t)\bar s(t)\alpha v_m(t)\beta(t)dt$. Aggregating the resource consumption over all degree classes in group $m$ and then aggregating over all groups we get the instantaneous rate of resource consumption for word-of-mouth strategy as
\begin{equation}
\sum_{m=1}^M \cancel{g_m} c_m(v_m(t)) \bar s(t)\alpha v_m(t)\beta(t) \sum_{k\in\mathbb K_m} \left\{\frac{p_k}{\cancel{g_m}}ki_k(t)\right\}. \label{eq:instantaneous_resource_rate_wom}
\end{equation}
Define $\bar i_m(t)\triangleq\sum_{k \in \mathbb K_m} ki_k(t)p_k$. Thus we get Eq. (\ref{eq:opt_prob_B_d}) by aggregating the resource expenditure over the complete campaign horizon. The constraint Eq. (\ref{eq:opt_prob_B_d}) is written from the perspective of a single (typical) node in the network; so, $B$ must be interpreted as a per-node expected budget.

The optimal control problem can now be stated as:
\small
\begin{subequations}
\label{eq:opt_prob}
\begin{align}
\underset{\boldsymbol u\in \boldsymbol U,\boldsymbol v\in \boldsymbol V}{\textnormal{max}} J & = \sum_{k\in \mathbb K} p_k i_k(T), \label{eq:cost_funtional}\\
\textnormal{s.t.:~~} \dot i_k(t) & =  \beta(t)s_k(t)k \left(\sum_{p=1}^M\sum_{l\in \mathbb K_p} q_{l} i_l(t)\underbrace{\alpha(1+v_p(t))}_{\textrm{word-of-mouth}}\right) \nonumber \\
&  + \underbrace{u_m(t)s_k(t)}_{\textnormal{direct recruitment}};~\forall k\in \mathbb K_m,1\leq m \leq M, \label{eq:opt_prob_states} \\
i_k(0) & =  i_0; ~~k\in \mathbb K, \label{eq:opt_prob_init_cond} \\
& \hspace{-3em}  \int_0^T \sum_{m=1}^M \Big\{g_mb_m(u_m(t))~+ \nonumber \\ 
& \alpha v_m(t)\beta(t) c_m(v_m(t)) \bar i_m(t) \bar s(t) \Big\} dt = B. \label{eq:opt_prob_B_d}
\end{align}
\end{subequations}
\normalsize
As stated earlier, $i_0$ is the seed for the epidemic. Since, $s_k(t)=1-i_k(t),~0\leq t\leq T,~k\in \mathbb K$, so we have $|\mathbb K|=\sum_{m=1}^M|\mathbb K_m|$ state variables (and not $2|\mathbb K|$). Also, $\alpha(1 + v_m(t))\leq 1$, $\forall t\in[0,T]$. To satisfy this condition, we choose $\alpha$ and $v_{max}$ such that, $\alpha(1 + v_{max})\leq 1$.

Define $B_{full}$ as the resource required to run the campaign with full intensity throughout the campaign horizon. It is worth noting that the interesting case occurs when $0<B<B_{full}$. In this case, the solution to the problem allocates the limited resource over the campaign period. We will not consider cases where $B>B_{full}$ in this paper. For the cases considered, notice that it will never be optimal to underutilize the budget, hence we have put equality (instead of the inequality $\leq$) in Eq. (\ref{eq:opt_prob_B_d}) without loss of generality.

In Problem (\ref{eq:opt_prob}), depending on the application and situation, the campaigner can decide upon the number of direct and word-of-mouth controls, $M$, she wants to apply and the degree classes in the $m$th group, $\mathbb K_m$, they would target. Also, in some cases, it may not be feasible to apply word-of-mouth control. For such cases, we set $v_m\equiv 0,~1\leq m\leq M$.

We emphasize the implication of the factor $\alpha$ and the word-of-mouth control in our model (biological SI epidemics have $\alpha = 1$, \emph{i.e.}, every infected individual is a potential spreader). Irrespective of the value of $\alpha$ used ($0<\alpha\leq 1$), message propagation is always probabilistic whenever individuals interact, depending upon whether or not the topic of interest came up during the meeting. In our model, only $\alpha$ fraction of infected individuals `try' to spread the message; the rest of them are uninterested to begin with, and do not attempt spreading. Even the non-spreading infected individuals have the message, so a higher level of word-of-mouth control may convert some of them to spreaders (which is a reasonable model) and when word-of-mouth incentive subsides, such individuals may again become uninterested (non-spreaders). Advertising for a product, or advocating for a (risky) crowdfunding investment means putting personal reputation at stake, so the level of word-of-mouth incentive is an important factor influencing the spreading behavior of an individual. Note that the classical Susceptible-Infected-Recovered model will not capture this phenomenon because once recovered, individuals will \emph{never} spread the message. This makes the SI model (with $\alpha$) more suitable for the applications considered in this paper.

\section{Existence of a Solution}
The proof of the existence of a solution to an optimal control problem is of practical importance because sometimes even reasonable looking problems do not admit their extrema; see examples 1.1 and 1.2 in \cite[Chap. III, Sec. 1]{fleming1975deterministic}. The usual method to show the existence of a solution to an optimal control problem is to use Filippov/Cesari's existence theorem \cite[Chap. III, Sec. 2 and 4]{fleming1975deterministic}. However, they are only applicable to systems where the controls affect the system linearly. Due to the isoperimetric constraint (\ref{eq:opt_prob_B_d}), for non-linear $b_m(.)$ and $c_m(.)$, neither of the theorems is applicable in our case. This observation is based on the fact that, in the equivalent formulation of Problem (\ref{eq:opt_prob}), Eq. (\ref{eq:opt_prob_B_d}) can be replaced by:
\begin{align*}
\dot h(t) =& \sum_{m=1}^M \Big\{g_mb_m(u_m(t)) + \alpha v_m(t)\beta(t) c_m(v_m(t)) \bar i_m(t) \bar s(t) \Big\};\\
h(0) =&~0,~h(T) = B;
\end{align*}
which has, for non-linear $b_m(.)$ and $c_m(.)$, a non-linear dependence on the controls. Note that this is the standard way of handling the isoperimetric constraints in optimal control problems. Due to the non applicability of the standard Filippov/Cesari existence theorems, we show the existence of a solution to the formulated problem using first principles.

We detail the steps to show the existence of a solution to Problem (\ref{eq:opt_prob}). We make use of the `Extreme Value Theorem' \cite[Theorem 4.16]{rudin1964principles} to show the existence. It says that a continuous function on a compact space attains its supremum at some point in the compact space. The steps involved are as follows (Theorem \ref{thm:existence}): (i) We first define a compact set $\boldsymbol{W}$ using the set of admissible controls $\boldsymbol U \times \boldsymbol V$ (defined in (\ref{eq:U}) and (\ref{eq:V})) and the isoperimetric constraint (\ref{eq:opt_prob_B_d}). (ii) We then show that the functional $J$ defined in (\ref{eq:cost_funtional}) is continuous at all elements of that compact space, $(\boldsymbol u,\boldsymbol v)\in \boldsymbol{W}$. (iii) Then we show that the constraint differential equation (\ref{eq:opt_prob_states}), (\ref{eq:opt_prob_init_cond}) has a solution at all elements $(\boldsymbol u,\boldsymbol v)\in \boldsymbol{W}$, thus the constraint is satisfied for all elements of the compact space $\boldsymbol{W}$.

Let $\boldsymbol i(t)=\big(i_{K_{min}}(t),...,i_{K_{max}}(t)\big)$ denote the state variable vector. The right hand side (RHS) of equation (\ref{eq:opt_prob_states}) is denoted by $f_k(\boldsymbol i(t),t)$ (note that $(\boldsymbol u, \boldsymbol v)$ is a function of $t$). Let $\boldsymbol f(\boldsymbol i(t),t)=\big(f_{K_{min}}(\boldsymbol i(t),t),...,f_{K_{max}}(\boldsymbol i(t),t)\big)$. In this section we use 1-norm for vectors. However, the result holds for any $p$-norm due to equivalence of vector $p$-norms, $p\geq 1$ \cite{horn2012matrix}.

\begin{lemma}
\label{thm:lipschitz}
The function $\boldsymbol f(\boldsymbol i(t),t)$ is Lipschitz continuous in $\boldsymbol i(t)$.
\end{lemma}
\begin{proof}
Notice that,
\small
\begin{align*}
& |\boldsymbol f(\boldsymbol i(t),t)-\boldsymbol f(\boldsymbol{\hat i}(t),t)|=\sum_{k\in \mathbb K}|f_k(\boldsymbol i(t),t)- f_k(\boldsymbol{\hat i}(t),t)| \\
& \leq \sum_{k\in \mathbb K}\bigg| \beta(t)k\alpha \Big[(1-i_k(t)) \sum_{m=1}^M\sum_{l\in \mathbb K_m} q_li_l(t)(1+v_m(t)) \\
& - (1-\hat i_k(t)) \sum_{m=1}^M\sum_{l\in \mathbb K_m} q_l\hat i_l(t)(1+v_m(t))\Big] \bigg| \\
& + \sum_{k\in \mathbb K}\left| u_m(t)\left( s_k(t)-\hat s_k(t) \right) \right|
\end{align*}
\begin{align*}
& \leq u_{max}|\boldsymbol i-\boldsymbol{\hat i}| \\
& + \beta_{max}K_{max}\alpha\sum_{k\in \mathbb K}\bigg| \sum_{m=1}^M\sum_{l\in \mathbb K_m}\Big( q_l(1+v_m(t))(i_l(t)-\hat i_l(t)) \\
& - q_l(1+v_m(t))\big(\underbrace{i_l(t)i_k(t)-\hat i_l(t) \hat i_k(t)}_{i_li_k-i_l\hat i_k+i_l\hat i_k-\hat i_l\hat i_k}\big) \Big) \bigg| \\
& \leq u_{max}|\boldsymbol i-\boldsymbol{\hat i}| + \beta_{max}K_{max}^2\alpha(1+v_{max}) |\boldsymbol i-\boldsymbol{\hat i}| \\
& + \beta_{max}K_{max}^2\alpha(1+v_{max})|\boldsymbol i-\boldsymbol{\hat i}|\times 2.
\end{align*}
\normalsize
Thus, $|\boldsymbol f(\boldsymbol i(t),t)-\boldsymbol{f}(\boldsymbol{\hat i}(t),t)|\leq C|\boldsymbol i(t)-\boldsymbol{\hat i}(t)|$ which establishes Lipschitz continuity of $\boldsymbol f(\boldsymbol i(t),t)$ in $\boldsymbol i(t)~\forall~(\boldsymbol u,\boldsymbol v) \in \boldsymbol{U} \times \boldsymbol{V}$.
\end{proof}

\begin{lemma}
\label{thm:continuous_dependence}
If $\boldsymbol{\tilde i(t)}$ is a solution of the system of ODEs (\ref{eq:opt_prob_states}) and (\ref{eq:opt_prob_init_cond}), then $\boldsymbol{\tilde i(t)}$ is a continuous function of $(\boldsymbol u,\boldsymbol v) \in \boldsymbol{U} \times \boldsymbol{V}$.
\end{lemma}
\begin{proof}
The continuity of $\boldsymbol{\tilde i}(t)$ can be shown by using the `Theorem on Continuous Dependence' \cite[pg. 145]{walter1998ordinary} of the solution of an ordinary differential equation on the vector field on the RHS. The theorem on continuous dependence states that if $\boldsymbol{\tilde i}(t), t\in[0,T]$ is a solution of (\ref{eq:opt_prob_states}),(\ref{eq:opt_prob_init_cond}); then given $\epsilon$ there exist $\delta$ such that $\big|\boldsymbol{\tilde i}(t)-\boldsymbol{\hat{i}}(t)\big| < \epsilon$, whenever $\big|\boldsymbol f(\boldsymbol i(t),t) - \boldsymbol{\hat{f}}(\boldsymbol i(t),t)\big| < \delta$, for $t\in[0,T]$. Here $\boldsymbol{\hat{i}}(t)$ is the solution of a perturbed version of (\ref{eq:opt_prob_states}),(\ref{eq:opt_prob_init_cond}), $\dot{i}_k(t)=\hat{f}_k(\boldsymbol i(t),t)$, $\forall k\in \mathbb K$, where $(\boldsymbol u, \boldsymbol v)$ in $\boldsymbol f(\boldsymbol i(t),t)$ is perturbed to $(\boldsymbol{\hat{u}},\boldsymbol{\hat{v}})$ to get $\boldsymbol{\hat{f}}(\boldsymbol i(t),t)$.

Notice that,
\small
\begin{align*}
& |\boldsymbol f(\boldsymbol i(t),t)-\boldsymbol{\hat f}(\boldsymbol i(t),t)|=\sum_{k\in \mathbb K}|f_k(\boldsymbol i(t),t)-\hat f_k(\boldsymbol i(t),t)| \\
& \leq \sum_{k\in \mathbb K}\left| \beta(t)s_k(t)k\alpha \sum_{m=1}^M\sum_{l\in \mathbb K_m} q_li_l(t)\left(v_m(t)-\hat v_m(t)\right) \right|\\
& + \sum_{k\in \mathbb K} \left| s_k(t)\left( u_m(t)-\hat u_m(t) \right) \right| \\
& \leq \beta_{max}K_{max}^2\alpha|\boldsymbol v - \boldsymbol{\hat v}| + K_{max}|\boldsymbol u - \boldsymbol{\hat u}|,
\end{align*}
\normalsize
where, maximum values $\beta_{max}$ for $\beta(t)$, $K_{max}$ for $k$, 1 for $s_k(t)$, 1 for $i_k(t)$, and 1 for $q_k$ are used in the last step (note that the excess degree distribution, $q_k$, is a probability mass function which is non-negative and sums to 1). So when,
\begin{align*}
& |\boldsymbol v - \boldsymbol{\hat v}|+|\boldsymbol u - \boldsymbol{\hat u}|\leq \frac{\delta}{\textnormal{max}\{\beta_{max}K_{max}^2\alpha, K_{max}\}}
\end{align*}
we have, $|\boldsymbol f - \boldsymbol{\hat f}| \leq \delta \Rightarrow |\boldsymbol {\tilde i}(t) - \boldsymbol{\hat i}(t)| \leq \epsilon, \forall t\in [0~T];$ which establishes the continuity of $\boldsymbol{\tilde i}(t)$ in $(\boldsymbol u, \boldsymbol v)$.

The `Theorem on Continuous Dependence' also requires Lipschitz continuity of $\boldsymbol f(.)$ in the state variable vector $\boldsymbol i(t)$, which follows from Lemma \ref{thm:lipschitz}.
\end{proof}

Define $\boldsymbol{W'}\triangleq \{(\boldsymbol u,\boldsymbol v):\int_0^T \sum_{m=1}^M \{g_mb_m(u_m(t)) + \alpha v_m(t) \beta(t) c_m(v_m(t)) \bar i_m(t) \bar s(t) \} dt = B\}$, where $g_m=\sum_{k \in \mathbb K_m} p_k$, $\bar i_m(t)=\sum_{k \in \mathbb K_m} ki_k(t)p_k$ and $\bar s(t)=\sum_{k\in \mathbb K} (kp_k/\bar k)s_k(t)$. Mean degree of the network is represented by $\bar k$ and $i_k(t)$'s are the solution to the initial value problem (IVP) (\ref{eq:opt_prob_states}), (\ref{eq:opt_prob_init_cond}).

\begin{lemma}
\label{thm:W_compact}
The set $\boldsymbol{W}$ is compact, where $\boldsymbol{W}\triangleq\boldsymbol{W'} \cap (\boldsymbol{U}\times \boldsymbol{V})$.
\end{lemma}
Note that the set $\boldsymbol{W}$ consists of all controls in $\boldsymbol{U}\times \boldsymbol{V}$ for which the budget constraint (\ref{eq:opt_prob_B_d}) is satisfied.
\begin{proof}
\emph{Step I: The sets $\boldsymbol U$ and $\boldsymbol V$ are compact:} Note that $U_m$ in (\ref{eq:U_m}) and $V_m$ in (\ref{eq:V_m}) are equicontinuous and equibounded sets and hence precompact (Arzela-Ascoli Theorem \cite[Theorem 1.6.3]{atkinson2009theoretical}). In addition, $U_m$ and $V_m$ are closed by definition. Hence $U_m$ and $V_m$ are compact sets. Now, $\boldsymbol U$ and $\boldsymbol V$ being cross products of finite number of compact sets are compact (Tychonoff's Theorem \cite[Page 392]{rudin1973functional}).

\emph{Step II: The set $\boldsymbol{W'}$ is closed:} From Lemma \ref{thm:continuous_dependence}, the solution to the IVP (\ref{eq:opt_prob_states}), (\ref{eq:opt_prob_init_cond}) is a continuous function of  $(\boldsymbol u,\boldsymbol v) \in \boldsymbol{U} \times \boldsymbol{V}$. Thus, $\bar i_m(t)$ and $\bar s(t)$ are continuous functions of $(\boldsymbol u,\boldsymbol v) \in \boldsymbol{U} \times \boldsymbol{V}$. Also, $b_m(),~c_m()$ are continuous in their arguments, hence $\boldsymbol{W'}$ can be expressed as $\boldsymbol{W'} = \{(\boldsymbol u,\boldsymbol v):\eta(\boldsymbol u,\boldsymbol v)=B\}$, where $\eta(.)$ is a continuous function in its arguments. Using standard techniques in real analysis, we can show $\boldsymbol{W'}$ is closed. Let $\{(\boldsymbol u_n,\boldsymbol v_n)\}$ be a sequence such that $(\boldsymbol u_n,\boldsymbol v_n)\in \boldsymbol{W'}~\forall n\in \mathbb Z^+$, the set of non-negative integers. Let the limit point of the sequence be $(\boldsymbol{\bar u},\boldsymbol{\bar v})$. Then, $\eta(\boldsymbol{\bar u},\boldsymbol{\bar v})=\eta(\lim_{n\rightarrow \infty} \boldsymbol{u}_n,\lim_{n\rightarrow \infty} \boldsymbol{v}_n)=\lim_{n\rightarrow \infty} \eta(\boldsymbol{u}_n, \boldsymbol{v}_n)=\lim_{n\rightarrow \infty} B = B$. Third to last equality follows from the fact that $\eta(.)$ is continuous in its arguments. Thus the limit point lies in $\boldsymbol{W'}$, which means $\boldsymbol{W'}$ is closed.

\emph{Step III: The set $\boldsymbol W$ is compact:} The product of compact sets $\boldsymbol{U} \times \boldsymbol{V}$ is compact (Tychonoff's Theorem \cite[Page 392]{rudin1973functional}). The intersection of closed and compact sets $\boldsymbol{W'} \cap (\boldsymbol{U}\times \boldsymbol{V}) \triangleq \boldsymbol{W}$ is compact \cite[Page 38]{rudin1964principles}.
\end{proof}

\begin{thm}
\label{thm:existence}
There exist optimal control signals $(\boldsymbol u^*, \boldsymbol v^*)\in \boldsymbol{U} \times \boldsymbol{V}$ and the corresponding solutions $\boldsymbol{i}(t)$ to the IVP (\ref{eq:opt_prob_states}), (\ref{eq:opt_prob_init_cond}) such that $(\boldsymbol u^*, \boldsymbol v^*)\in \underset{(\boldsymbol u, \boldsymbol v) \in \boldsymbol{U} \times \boldsymbol{V}}{\textnormal{ argmax }} \{J(\boldsymbol u, \boldsymbol v)\}$ in Problem (\ref{eq:opt_prob}).
\end{thm}
\begin{proof}
In light of the discussion earlier in the section, this theorem can be proved in three steps:

\emph{Step I: The set $\boldsymbol{W}$ is compact:} Lemma \ref{thm:W_compact}.

\emph{Step II: Initial value problem (\ref{eq:opt_prob_states}),(\ref{eq:opt_prob_init_cond}) has a solution for any $(\boldsymbol u,\boldsymbol v) \in \boldsymbol{W}$:} The solution to the initial value problem (\ref{eq:opt_prob_states}),(\ref{eq:opt_prob_init_cond}) (aggregated for all $k\in \mathbb K$) exists because $\boldsymbol f(\boldsymbol i(t),t)$ is Lipschitz continuous in $\boldsymbol i(t)~\forall~(\boldsymbol u,\boldsymbol v) \in \boldsymbol{U} \times \boldsymbol{V} \supseteq \boldsymbol{W}$ \cite[pg. 185]{birkhoff1989ordinary}. Lipschitz continuity follows from Lemma \ref{thm:lipschitz}.

\emph{Step III: The functional $J=\sum_{k\in\mathbb K}p_ki_k(T)$ in (\ref{eq:cost_funtional}) is continuous at all $(\boldsymbol u,\boldsymbol v) \in \boldsymbol{W}$:} From Lemma \ref{thm:continuous_dependence}, if $\boldsymbol{i}(t)$ is the solution of IVP (\ref{eq:opt_prob_states}), (\ref{eq:opt_prob_init_cond}), then $\boldsymbol i(T) = \boldsymbol i(t)|_{t=T}$ is continuous in $(\boldsymbol u, \boldsymbol v) \in \boldsymbol{U} \times \boldsymbol{V}$. But $\boldsymbol{W}$ is a compact subset of $\boldsymbol{U} \times \boldsymbol{V}$.
\end{proof}

\section{Solution via Non-linear Programming}

Once we have established the existence of a solution to Problem (\ref{eq:opt_prob}), we proceed to calculate it. We solved Problem (\ref{eq:opt_prob}) by using algorithms that are available to solve optimal control problems by converting them into discrete non-linear optimization problems \cite{becerra2004solving}. Matlab's non-linear optimization solver \texttt{fmincon()} is used to solve the discretized version of Problem (\ref{eq:opt_prob}). The function solves the constrained non-linear optimization problems and requires the objective function and the constraints as inputs. We brief the discretization steps in the following.

The campaign duration $[0,T]$ is sampled at $N+1$ equidistant time-points, such that $t_0=0$ and $t_N=T$. Let $i_k(t_n), u_m(t_n)$ and $v_m(t_n)$ be the values of the state and control variables at those instances, $n=0,1,...,N$; and let $\Delta t=t_1-t_0=...=t_{N}-t_{N-1}$. The idea is to keep $N$ large enough to achieve the desired accuracy and compute the values of $u_m(t_n)$ and $v_m(t_n)$, $\forall n,~\forall m$ using the optimization routine.

\emph{Objective function:} The objective to be maximized is $\sum_{k\in \mathbb K} p_ki_k(t_N)$. However, to evaluate this, we need to solve the initial value problem (IVP) (\ref{eq:opt_prob_states}) and (\ref{eq:opt_prob_init_cond}). We do so by following Heun's method \cite[Sec. 1.1]{kutz2005practical} to solve the IVP. The global error due to discretization is bounded by $O(\Delta t^2)$ for this method.

Let the right hand side of Eq. (\ref{eq:opt_prob_states}) be denoted by $f_k(\boldsymbol i(t),t)$, where $\boldsymbol i(t)$ is vector of all $i_k(t)$; notice that $s_k(t)=1-i_k(t)$ and $u_m(t), v_m(t)$ are just functions of $t$. At the initial time-point, $i_k(t_0)=i_0,~\forall k$. The value of $i_k(t_n)$ for $1\leq n\leq N$, for a given $k$, is computed using both the left and right derivatives as follows:
\small
\begin{align}
i_k(t_n)  = & i_k(t_{n-1}) + \frac{\Delta t}{2} \Bigg[ f_k\big(\boldsymbol i(t_{n-1}),t_{n-1}\big) \nonumber \\
& + \overbrace{f_k\bigg(\underbrace{\boldsymbol i(t_{n-1}) + \Delta t f_k\big(\boldsymbol i(t_{n-1}),t_{n-1}\big)}_{\textnormal{approximation of } \boldsymbol i(t_n)},t_n\bigg)}^{\textnormal{approximate right derivative}} \Bigg]. \nonumber
\end{align}
\normalsize
The value $\boldsymbol i(t_N)$, at the last time-point is used in the objective function.

\emph{Constraints:} The following constraints are fed as inputs to the routine: 
\begin{enumerate}
\item The inequality constraints, $u_m(t_n)\ge 0$, $v_m(t_n)\ge 0$, $u_m(t_n)\le u_{max}$, $v_m(t_n)\le v_{max}$, $\forall n,~\forall m$.
\item The budget constraint (\ref{eq:opt_prob_B_d}) leads to the equality constraint: $\sum_{n=1}^N\sum_{m=1}^M \Big\{ g_mb_m(u_m(t_n)) + \alpha v_m(t_n) \beta(t_n) c_m(v_m(t_n)) \bar i_m(t_n) \bar s(t_n) \Big\}\Delta t - B = 0$. Here $g_m=\sum_{k\in \mathbb K_m} p_k$, $\bar i_m(t_n)=\sum_{k \in \mathbb K_m} ki_k(t_n)p_k$, and $\bar s(t_n)=\sum_{k\in \mathbb K} (kp_k/\bar k)s_k(t_n)$. Values of $i_k(t_n),~\forall n,~\forall k$, are obtained from the computation above and $s_k(t_n)=1-i_k(t_n)$.
\end{enumerate}

The computation in the optimization routine is initialized with some initial guess for $u_m(t_n)$ and $v_m(t_n)$, $\forall n,~\forall m$ and the routine refines it until the stopping criteria are met. The function \texttt{fmincon()} uses a combination of factors, like change in the objective function value, change in the values of the variables being optimized, magnitude of the gradient etc. to decide the stopping criteria.

We have used Heun's method over other possibilities such as Euler and Runge-Kutta methods. For reasonable values of $N$ (so that the number of optimization variables is not too large), Euler's method was numerically unstable for certain parameter values (\emph{e.g.} high $\beta$). This is so because Euler's method is only $O(\Delta t)$ accurate. The Runge-Kutta method requires the value of the optimization variables at $2\times N$ time points for $N$ point discretization \cite[Sec. 1.1]{kutz2005practical}. Although more accurate for small systems (with small values of $M$), the memory and execution time requirements (in the optimization routine) for larger systems will increase non-linearly, which makes its use unsuitable. Heun's method offers a compromise between these factors (accuracy and memory requirement/execution time) and works well in practice.

\section{Networks used in this Study}

\begin{figure}[ht!]
\centering
\includegraphics[width=55mm]{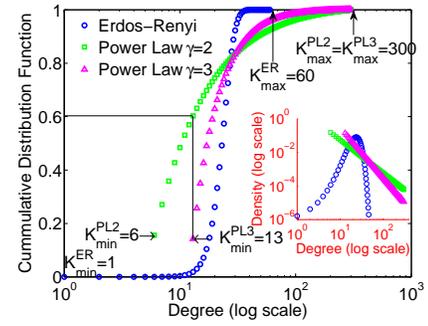}
\caption{\small{Cumulative degree distribution and degree distribution (inset) (defined in Sec. \ref{sec:sys_model}) of the networks used in this paper.}}
\label{fig:degree_distribution}
\end{figure}

We have used degree distributions from three networks to demonstrate the results. The first one is an Erd\H os-R\'enyi network, which is known to have a Poisson degree distribution. The probability of finding a node with degree $k$ is, $p_k = e^{-\lambda}\lambda^k/k!,~\forall k\in \mathbb K$, where $\lambda$ is the mean degree in the network and $k!$ denotes the factorial of the integer $k$. The other two are scale free networks, which are known to follow the power law degree distribution, $p_k = ck^{-\gamma}, \forall k \in \mathbb K$, where $c$ is a properly chosen scalar to normalize the degree distribution to $1$. Here $\gamma$ is the exponent of the power law and lies between $2$ and $3$ for most real networks (including social networks) \cite{newman2010networks}. 
Hence we have chosen $\gamma=2$ and $\gamma=3$ for the two scale free networks used in this study. We name the networks ER, PL2 and PL3 in the rest of the paper.

The three networks are chosen such that the mean degrees for all three of them are almost the same. They cannot be made exactly the same because the minimum and maximum degrees for the networks are discrete quantities. The minimum and maximum degrees for the scale free networks are $K_{min}^{PL2}=6$, $K_{max}^{PL2}=300$, $K_{min}^{PL3}=13$, $K_{max}^{PL3}=300$, which yield the mean degree for the two networks as $\bar k_{PL2}=22.47$, $\bar k_{PL3}=24.03$. The mean degree for the Erd\H os-R\'enyi network is set to $\bar k_{ER}=\lambda = 23.60$, with minimum and maximum degrees as $K_{min}^{ER}=1$, $K_{max}^{ER}=60$.

The maximum degrees for both the scale free networks are the same, so that the hubs (which will aid in information spreading) have similar degrees in both the cases. We preferred to keep the mean degree the same for all the three networks because, for any given (large) size, all three of them will have (almost) the same number of links. Thus, none of the networks will have a statistical advantage in spreading the information, which propagates through the links of the network. The probability distribution and density functions of the degrees of the three networks are shown in Fig. \ref{fig:degree_distribution}.

\section{Results and Discussions}

In this section, we first explain the values for the model parameters and validate the degree based compartmental model for the SI process for configuration model networks. Then we present the results demonstrating the importance of timing the incentives properly, and identify the degree classes which are more useful to target for maximum spreading. Then we study the effect of various model parameters on the reward functional, $J$ and compare the optimal strategy with static and bang-bang control strategies. In the static control strategy, the direct and word-of-mouth controls are implemented at $\kappa\leq 1$ times the maximum allowed value in all degree classes. The value of $\kappa$ is selected so that this strategy respects the same budget constraint as the optimal control strategy. In the bang-bang strategy, direct and word-of-mouth controls are implemented in all degree classes at the maximum strength till the resource lasts. As will be seen in this section, control strengths are strongest in the early stages of the campaign and gradually subside; this motivates the bang-bang strategy.

\subsection{Default Model Parameters and Validation of Degree Based Compartmental Model}
\label{sec:defalut_parameters}

\begin{figure}[ht!]
\centering
\includegraphics[width=53mm]{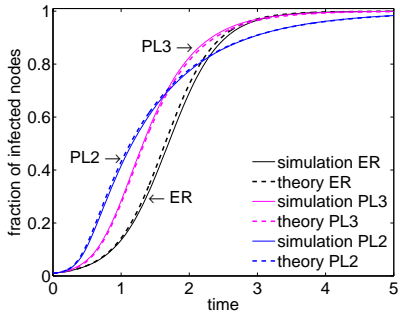}
\caption{\small{Validation of the degree based compartmental model for an SI epidemic on $10^4$ node configuration model networks. We average over 20 simulations. $\alpha=1$, $\beta=0.12$, $i_0=0.01$.}}
\label{fig:sim_vs_theory_SI_config_model}
\end{figure}

\emph{Model Parameters:} Unless otherwise stated, the model parameters are set to the following values for the results reported in the rest of this section. The fraction of infected nodes (relative to the total population) at the start of the epidemic, $i_0$, is set to $0.01$. The small value of $i_0$ implies that we have started campaigning at the early stages of the epidemic. 

The campaign deadline and spreading rate together decide the extent of spreading in the uncontrolled system. We have normalized the campaign deadline to $T=1$ and the spreading rate is set to $\beta=0.12$ (for cases where the spreading rate is constant over the campaign period). We choose such a value of $\beta$ because it leads to low to moderate spreading in the uncontrolled system ($i(T)=0.040, 0.058$ and $0.126$ for ER, PL3 and PL2 networks respectively, for $i_0=0.01$), the situation which would require campaigning. 

The value of $\alpha$, the fraction of the infected population interested in further spreading the message is set to $0.5$. The maximum value of direct control is set to $u_{max}=0.12$, equal in magnitude to the spreading rate, irrespective of whether word-of-mouth control is used or not. The maximum value for the word-of-mouth control is set to $v_{max}=0.5$; thus, we can only increase the fraction of spreaders by a maximum of $50\%$ by applying word-of-mouth control, which seems reasonable.

The functions deciding the cost of applying the direct and word-of-mouth controls are, $b_m(u(t)) = \hat b_m u^2(t)$ and $c_m(v(t)) = d\hat c_m v^2(t)$, with $\hat b_m=\hat c_m = 1,~\forall m$. That is, the default value of convincing all the degree class groups (captured by $\hat b_m$ and $\hat c_m$) is the same. The parameter $d$ captures the relative cost of using word-of-mouth control over direct control and is set to $0.5$.

The budget, $B$ is set to $u_{max}^2\times T/8$, which is one-eighth of the value of the resource spent if direct control is used at full intensity and no word-of-mouth control is used throughout the campaign duration (for $\hat b_m=\hat c_m = 1,~\forall m$). This corresponds to the scarce resource case, which is likely to be encountered in most real situations. We are not interested in cases where the resource is so abundant that maximum control strength can be applied throughout the campaign duration. All the computations are carried out by discretizing the campaign duration into 51 time points.

\emph{Model Validation:} The comparison between system evolutions, captured by $i(t)=\sum_{k\in\mathbb K}i_k(t)p_k$, produced by the SI epidemic process simulated on configuration model networks, and that produced by differential equations from degree based compartmental models are shown in Fig. \ref{fig:sim_vs_theory_SI_config_model}. The simulation results are averaged over 20 runs on networks of size $10^4$ for all three types of networks---ER, PL3 and PL2. We see an excellent match between simulation and degree based compartmental model. To construct a configuration model network, sample the degree of all $10^4$ nodes from $p_k,~k\in \mathbb K$. Then select a half edge of any node and pair it with any other half edge available in the network. Self and multiple loops can occur, but their density goes to zero for large networks \cite{newman2010networks}. Also, if the last half edge is left unpaired, it is ignored.

\subsection{Timing the Incentives and Important Degree Classes}
\label{sec:results_timing_incen_imp_deg}

\subsubsection{Constant Spreading Rate, $\beta(t)=\beta$}
\label{sec:results_constant_beta}

For this result, the degree classes are divided into three groups ($M=3$), each targeted by a direct and a word-of-mouth control. For a given network, nodes in degree classes $K_{min}$ to $\tilde k_1$ are assigned to the first group, those in degree classes $\tilde k_1 +1$ to $\tilde k_2$ are assigned to the second group and $\tilde k_2+1$ to $K_{max}$ are assigned to the third group, such that $\sum_{k=K_{min}}^{\tilde k_1} p_k \approx \sum_{k=\tilde k_1+1}^{\tilde k_2} p_k \approx \sum_{k=\tilde k_2+1}^{K_{max}} p_k \approx 1/3$. Thus, the three groups are selected such that roughly one-third of the nodes fall in each of the groups for all the three networks. This approach of creating the three groups remains the same in the rest of the paper. This is just an example of forming groups that is used to demonstrate results in this paper. Our method works for any disjoint collection of degree classes grouped to form $\mathbb{K}_m$'s. The degree classes are grouped only for the purpose of controlling the network, information diffusion dynamics is still captured with $|\mathbb K|$ equations, as suggested by Eq. (\ref{eq:opt_prob_states}). Following this approach, groups in the ER network have degree classes $\mathbb K_1^{ER}=\{1,...,21\}$, $\mathbb K_2^{ER}=\{22,...,25\}$, $\mathbb K_3^{ER}=\{26,...,60\}$; PL3 network: $\mathbb K_1^{PL3}=$\{13,14,15\}, $\mathbb K_2^{PL3}=$\{16,...,21\}, $\mathbb K_3^{PL3}=$\{22,...,300\}; and PL2 network: $\mathbb K_1^{PL2}=$\{6,7,8\}, $\mathbb K_2^{PL2}=$\{9,...,15\}, $\mathbb K_3^{PL2}=$\{16,...,300\}. 

Shapes of the direct and word-of-mouth optimal control signals for ER, PL3 and PL2 networks and the corresponding rates at which resource is allocated in each of the three groups are shown in Fig. \ref{fig:controls_const_beta}. To plot the resource allocation rate in group $m$, we have plotted $g_mb_m(u_m(t))$ and $\alpha v_m(t) \beta(t) c_m(v_m(t))\bar i_m(t) \bar s(t)$ over time for direct and word-of-mouth controls (as derived in (\ref{eq:instantaneous_resource_rate_wom})). The word-of-mouth control is representative of the cash-back or discount on membership renewals, announced for present customers if they introduce a friend to the service/product. The corresponding resource allocation rate is akin to the rate at which money is spent due to announcement of such cash-backs. As seen from the figure, for all three networks, both direct and word-of-mouth controls have larger strength at the early stages of the epidemic and they gradually decay. This is so because for the SI epidemic, early infection aids in faster propagation of the epidemic, 
hence the reward for an intense control, in-spite of the higher costs, is worthwhile. 

As seen from the resource allocation rate plots (Figs. \ref{fig:resource_ER_const_beta}, \ref{fig:resource_PL3_const_beta}, \ref{fig:resource_PL2_const_beta}), when all the degree classes are equally costly to influence, for PL3 and PL2 networks (which are more heterogeneous than the ER network), the highest group ($m=3$) is most important in the optimal strategy (which assigns highest resource for both types of controls to them; at least at the beginning stages in the case of direct controls, and always for the word-of-mouth controls). For the ER network, the medium group ($m=2$) attracts the most investment, which makes them the most important group for information spreading (more important than even the highest group, which is somewhat surprising). The group with highest degree classes is next most important, followed by the group with lowest degrees classes, which is as expected. The percentages of total (direct plus word-of-mouth) resource taken up by low, medium and high groups in the case of ER are 21\%, 47\%, 32\%; for PL3: 8\%, 29\%, 63\%; and PL2: 5\%, 17\%, 78\% respectively.

\emph{Explanation of this (`surprising') bahavior:} Information diffuses through edges in the network. We calculate the mean degree of nodes in the given network and given degree class group using $\bar k_m=\sum_{k\in \mathbb K_m}(kp_k)/\sum_{k\in \mathbb K_m}p_k$. The values for ER are $\bar k_1^{ER}=18.3,~\bar k_2^{ER}=23.5,~\bar k_3^{ER}=28.9$. For PL3: $\bar k_1^{PL3}=13.9,~\bar k_2^{PL3}=18.0,~\bar k_3^{PL3}=40.1$. And, for PL2: $\bar k_1^{PL2}=6.8,~\bar k_2^{PL2}=11.3,~\bar k_3^{PL2}=48.5$. Thus we see that the highest groups in heterogeneous networks (PL2 and PL3) have disproportionate advantage in terms of the number of links/edges to other nodes. Although this aids in both spreading and receiving information, targeting low/medium group nodes early on in the campaign is not useful because a randomly chosen node in them will most likely be a part of a collection of nodes connected to other low/medium group nodes and is most likely away from a hub. In PL2/PL3 networks, hubs (high degree nodes) are connected to low/medium degree nodes, but the rest of low/medium group nodes connect among themselves and the information penetrates slowly there. Later on in the campaign when we have exhausted high group nodes, it makes sense to target medium and low group nodes.

On the other hand, connectivity is more uniform in the ER network. Nodes in the highest group will be able to get information from medium and lower group nodes. In addition, directly targeting medium and lower groups increases the fraction of infected nodes in these groups, thus having a direct effect on increase in the objective (reward) functional (net fraction of infected nodes at the deadline).

The resource allocation rate plots (Figs. \ref{fig:resource_ER_const_beta}, \ref{fig:resource_PL3_const_beta}, \ref{fig:resource_PL2_const_beta}) also show that: (i) The word-of-mouth strategy assumes increasing importance with network heterogeneity (PL2$>$PL3$>$ER). The percentages of total resource allocated to word-of-mouth strategies (all 3 groups combined) for PL2, PL3 and ER networks are 44\%, 28\% and 19\% respectively. The more heterogeneous the network, the more hubs it has, thus more word-of-mouth resource is allocated to the highest group to infect the nodes attached to the hubs (see the curves corresponding to $m=3$ in the figures for PL2, PL3 networks). The ER network does not have a group with distinct advantage in terms of degrees, so the word-of-mouth resource allocation is low. (ii) The more heterogeneous the network, the more disparity there is in the allocated resources for high, medium and low groups for both control strategies. Thus, high degree nodes in more heterogeneous networks have more relative importance than their counterparts in less heterogeneous networks.

To achieve the above resource allocation, the optimal strategy uses controls as shown in Figs. \ref{fig:control_ER_const_beta}, \ref{fig:control_PL3_const_beta} and \ref{fig:control_PL2_const_beta}. The direct controls follow the same trend as the corresponding resource allocation rate curves; however, the word-of-mouth controls show different trends. For the ER network, more word-of-mouth rewards (the controls in form of cash-backs, for example) are announced for the medium group, followed by the low and then the highest group. For the heterogeneous networks, highest reward is announced in medium group, followed by high and low groups. Since the highest group has more nodes around it, the resource spent due to the group's cash-back claims is more than that of the other groups in the heterogeneous network. More rewards are announced in the medium group to incentivize their members for strong spreading to their neighbors, as there are fewer neighbors in the group.

Table \ref{table:time_taken} gives an idea of the time taken to solve the optimal control problem for various networks for the parameter values chosen in Sec. \ref{sec:defalut_parameters}. Note that this time is dependent on the values chosen for the parameters (which is true for any numerically solved optimization problem).

\begin{table}[ht]
\caption{\small{Time taken (in sec, on a machine with Intel quard-core i3-2100 CPU @ 3.10 GHz, 4 GB RAM, and running a 64 bit operating system) to solve Problem (\ref{eq:opt_prob}) for various networks for different values of $M$ (averaged over 10 runs) when $[0,T]$ is discretized to $N=51$ time points.}}
\label{table:time_taken}
\centering
\begin{tabular}{l l l l l l l}
 & $M=1$ & $M=2$ & $M=3$ & $M=4$ & $M=5$ & $M=10$ \\
\hline
ER & 8.50 & 45.96 & 77.99 & 185.68 & 284.38 & 1002.75 \\
PL3 & 12.25 & 85.52 & 143.33 & 244.34 & 341.58 & 1217.80 \\
PL2 & 12.36 & 94.64 & 169.42 & 251.97 & 395.72 & 1222.33 \\
\hline
\end{tabular}
\label{table:definition_parameters}
\end{table}

\begin{figure*}[ht!]
\centering
\subfloat[Controls, ER, $\beta(t)=\beta,~\forall t$. \label{fig:control_ER_const_beta}]{
\includegraphics[width=53mm]{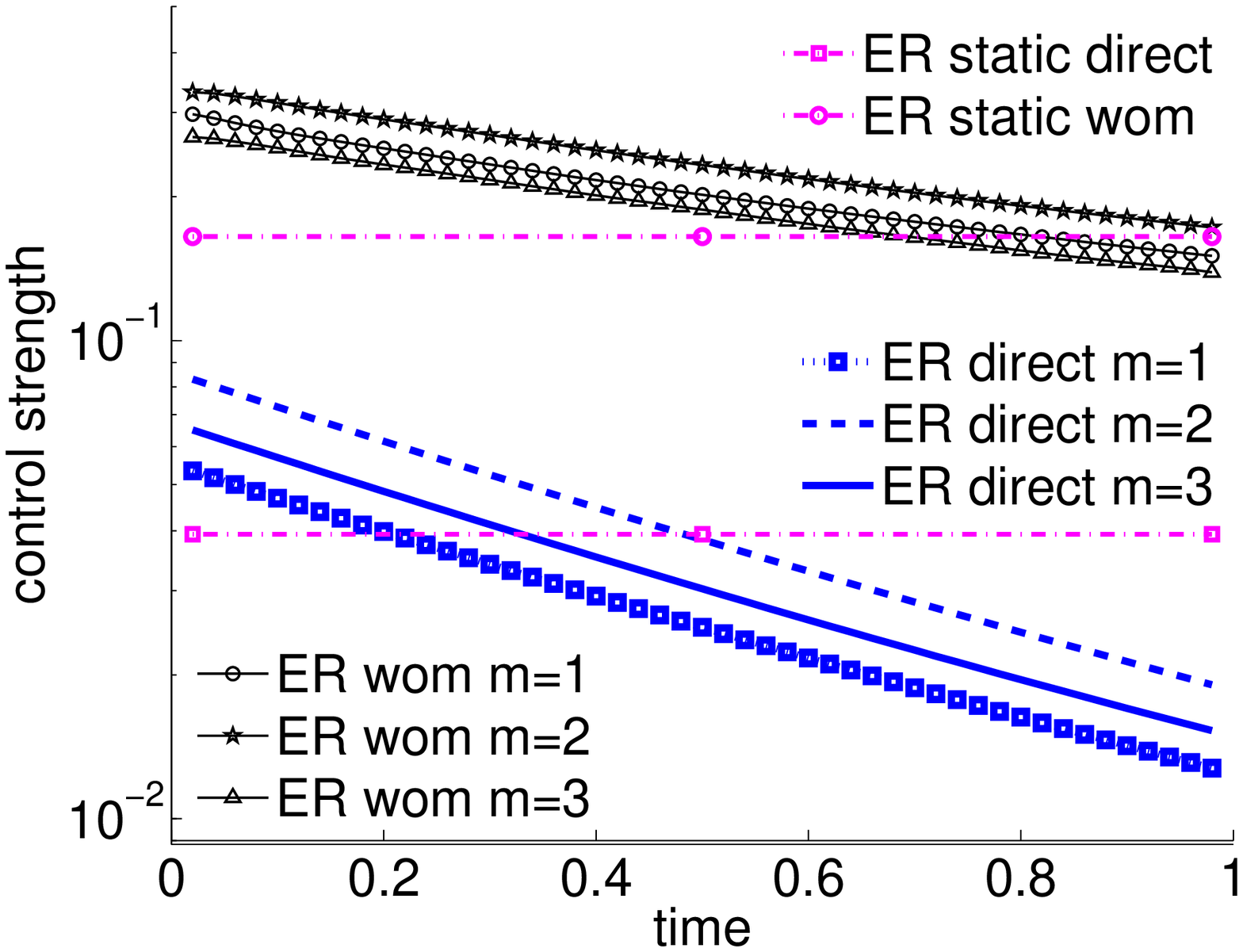} }
\hfill
\subfloat[Controls, PL3, $\beta(t)=\beta,~\forall t$. \label{fig:control_PL3_const_beta}]{
\includegraphics[width=53mm]{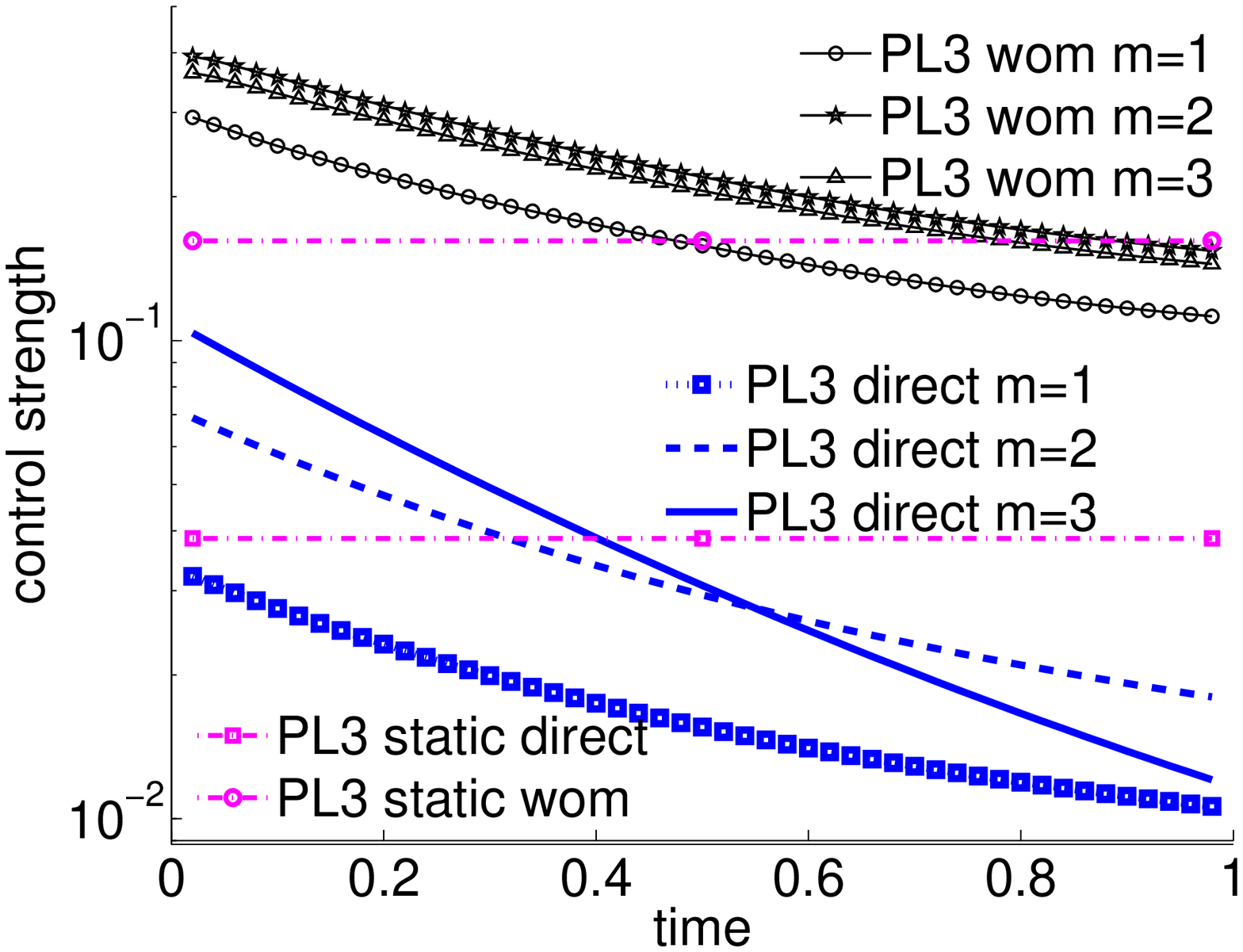} }
\hfill
\subfloat[Controls, PL2, $\beta(t)=\beta,~\forall t$. \label{fig:control_PL2_const_beta}]{
\includegraphics[width=53mm]{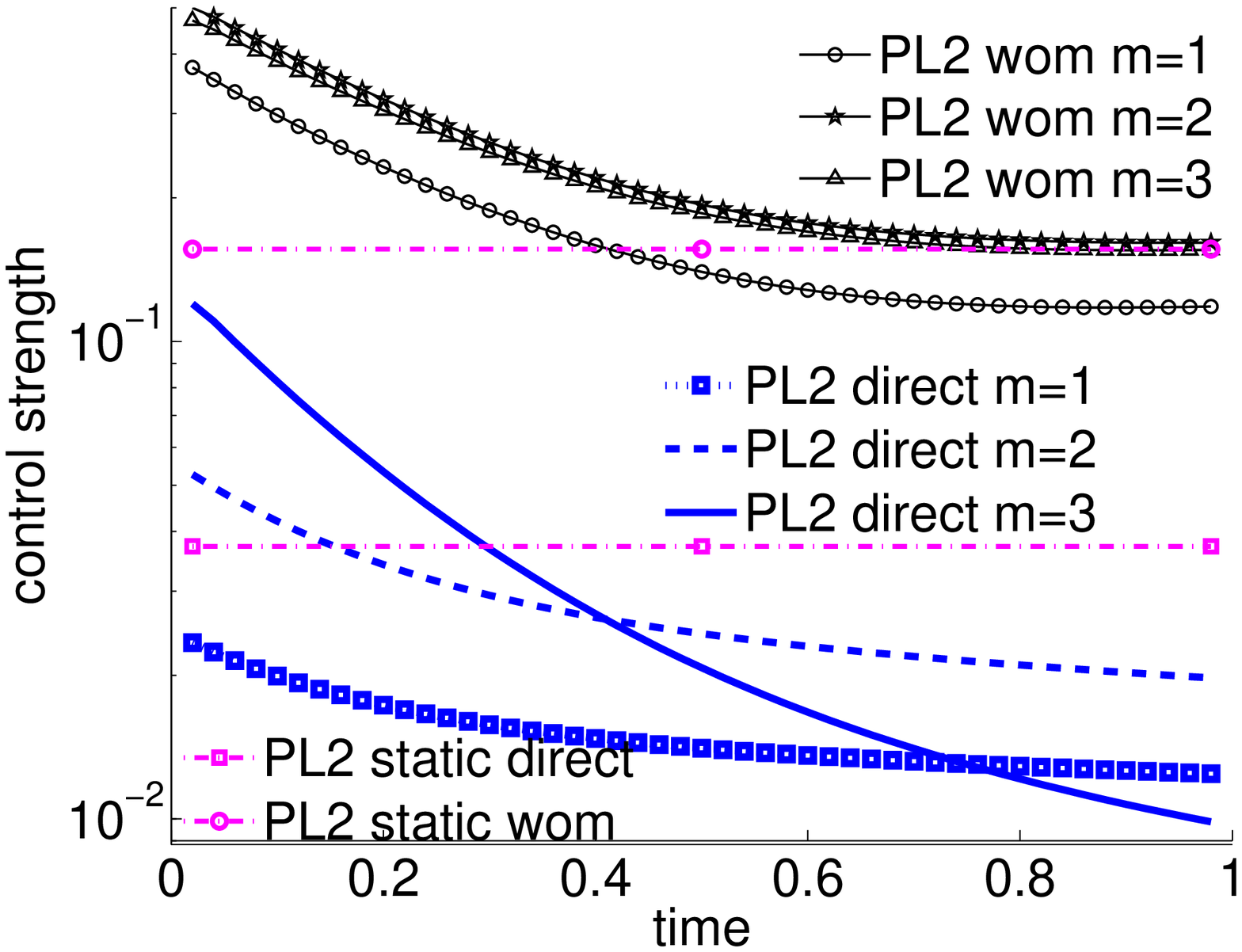} }
\hfill
\subfloat[Resource allocation, ER, $\beta(t)=\beta,~\forall t$. \label{fig:resource_ER_const_beta}]{
\includegraphics[width=53mm]{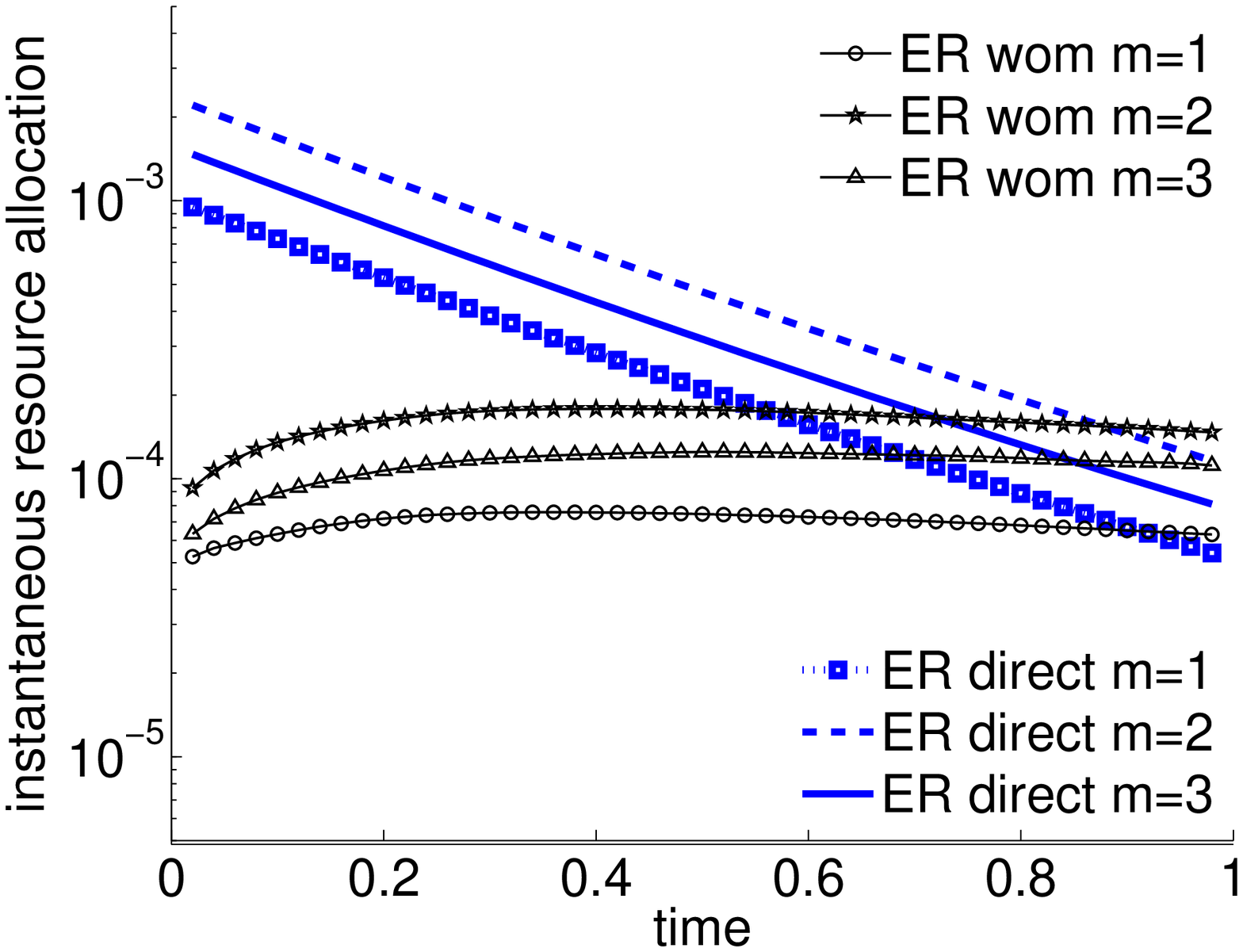} }
\hfill
\subfloat[Resource allocation, PL3, $\beta(t)=\beta,~\forall t$. \label{fig:resource_PL3_const_beta}]{
\includegraphics[width=53mm]{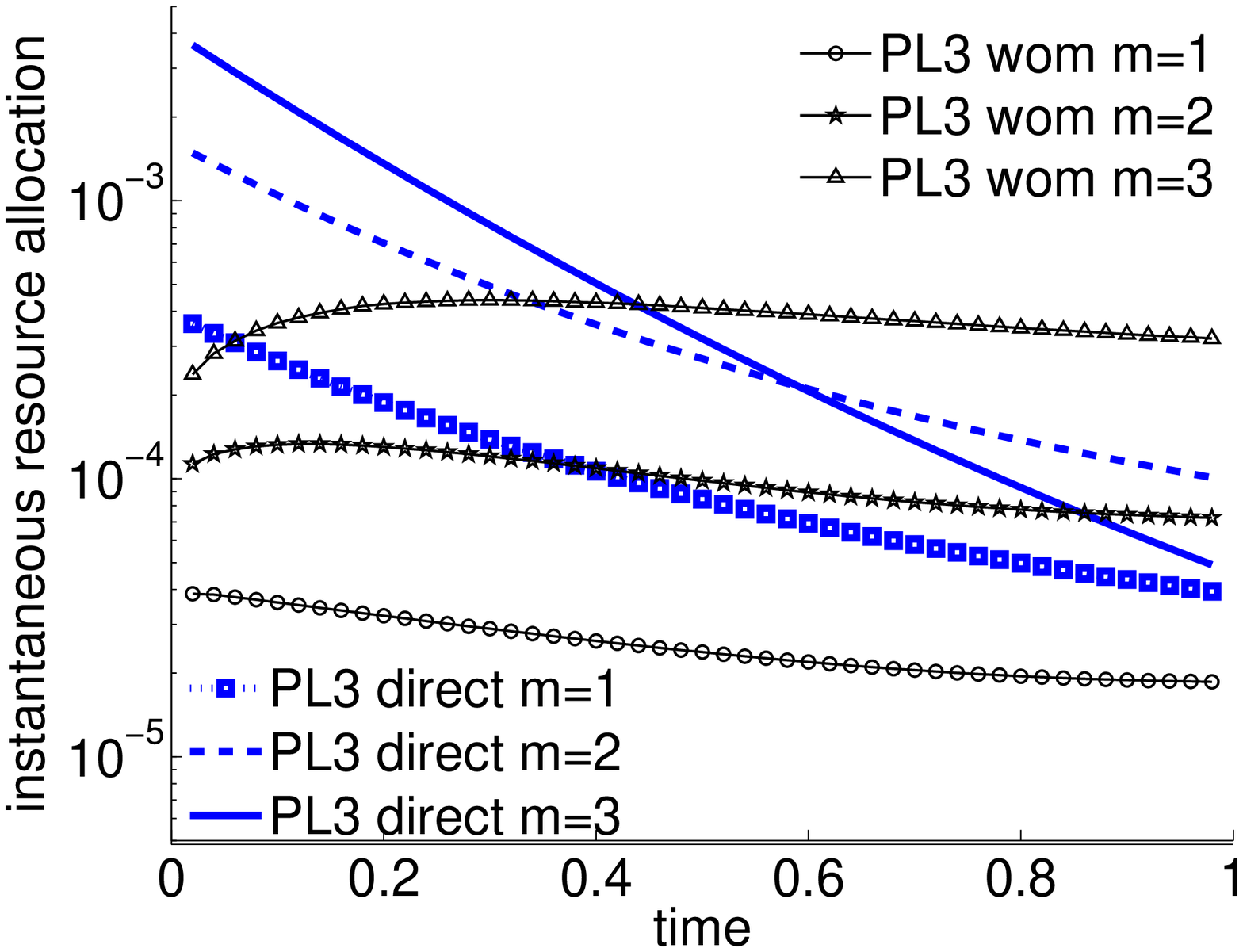} }
\hfill
\subfloat[Resource allocation, PL2, $\beta(t)=\beta,~\forall t$. \label{fig:resource_PL2_const_beta}]{
\includegraphics[width=53mm]{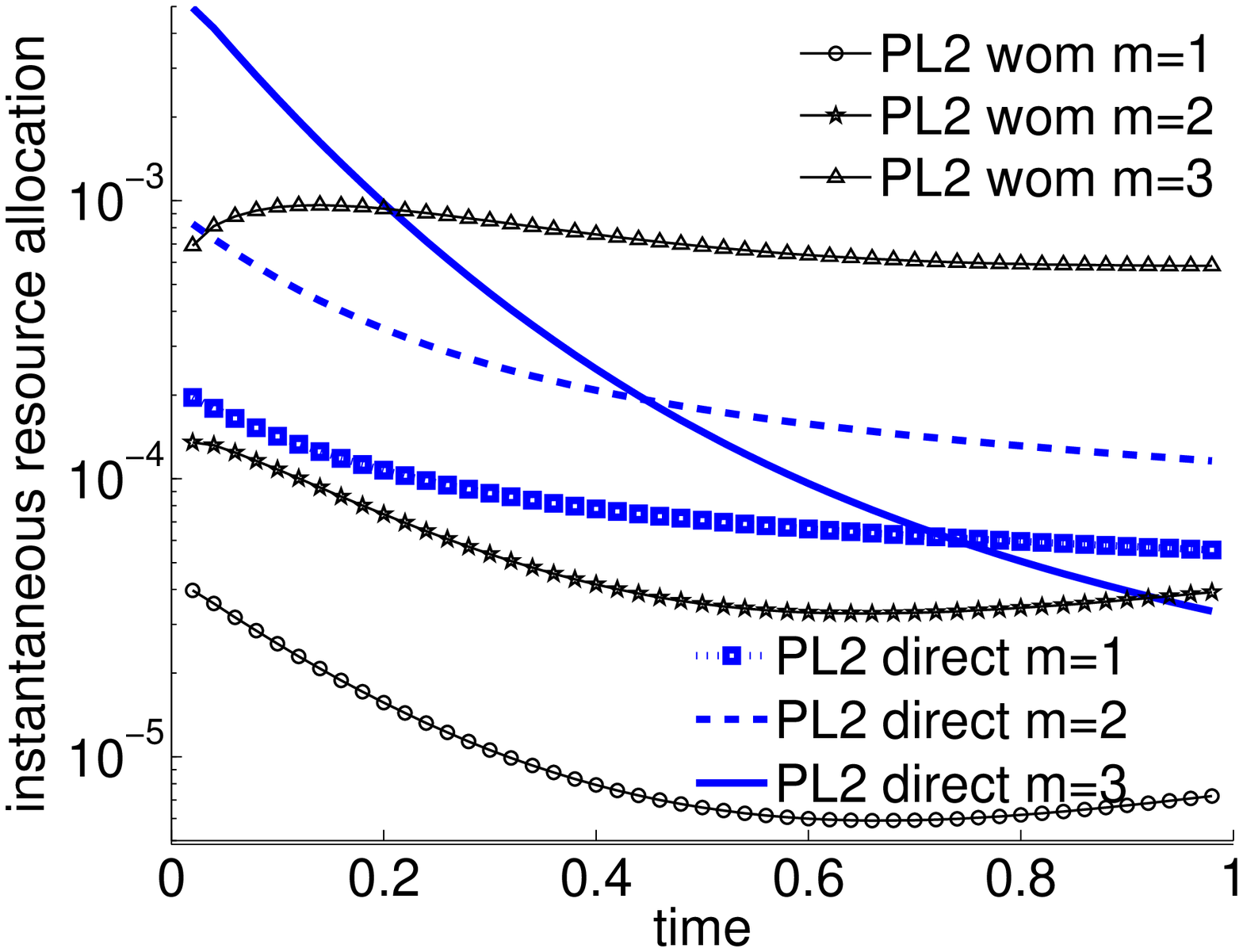} }
\caption{\small{Controls and resource allocation rate for $M=3$ for the case of constant spreading rate. Parameter values: $T=1$, $\beta=0.12$, $i_0=0.01$, $\alpha=0.5$, $u_{max}=0.12$, $v_{max}=0.5$, $\hat b_m=\hat c_m=1~\forall m$, $B = u_{max}^2T/8$, $d = 0.5$.}}
\label{fig:controls_const_beta}
\end{figure*}

\begin{figure*}[ht!]
\centering
\subfloat[Controls, PL3, $\downarrow \beta_1(t)$. \label{fig:control_PL3_decreasing_beta}]{
\includegraphics[width=53mm]{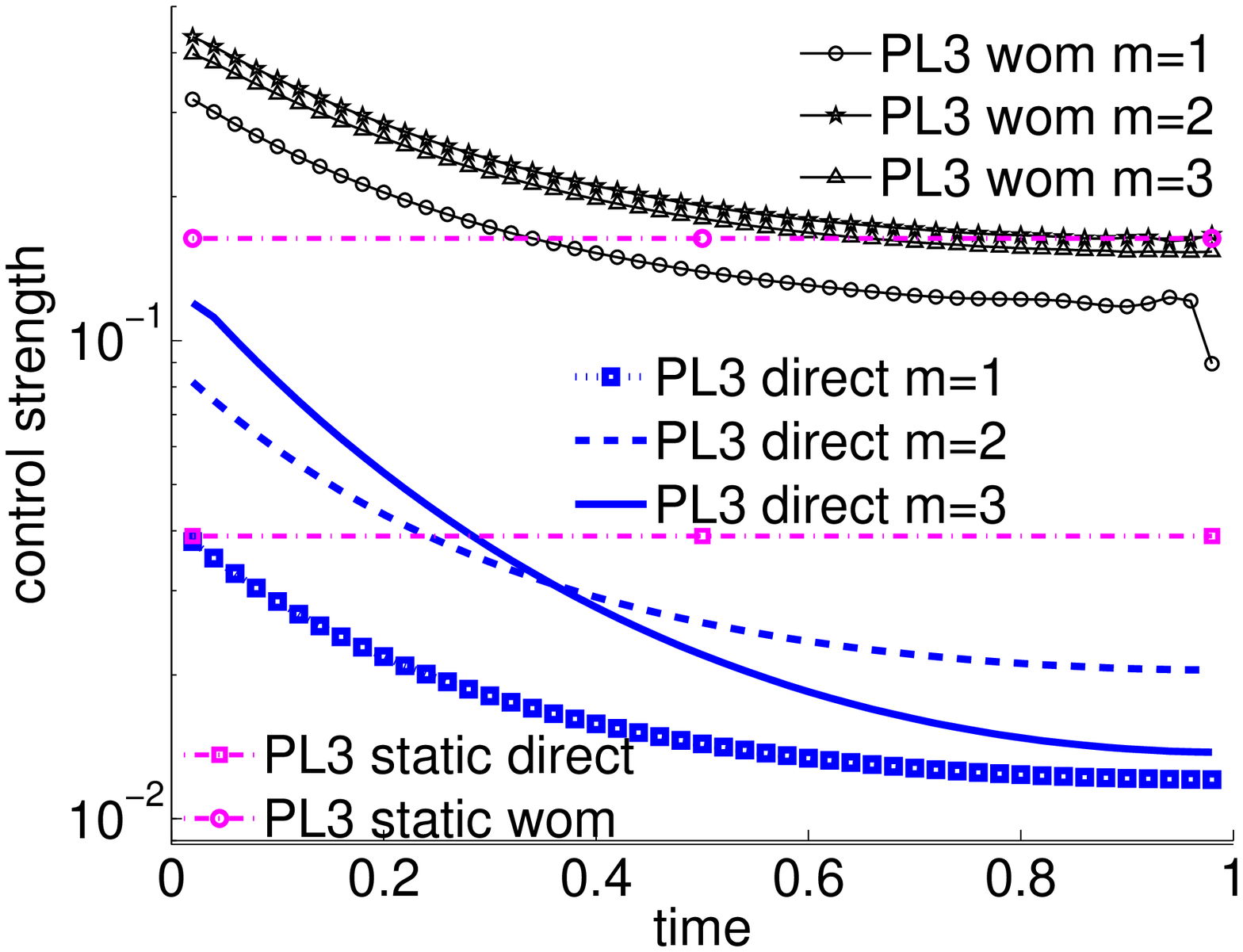} }
\hfill
\subfloat[Controls, PL3, $\uparrow \beta_2(t)$. \label{fig:control_PL3_increasing_beta}]{
\includegraphics[width=53mm]{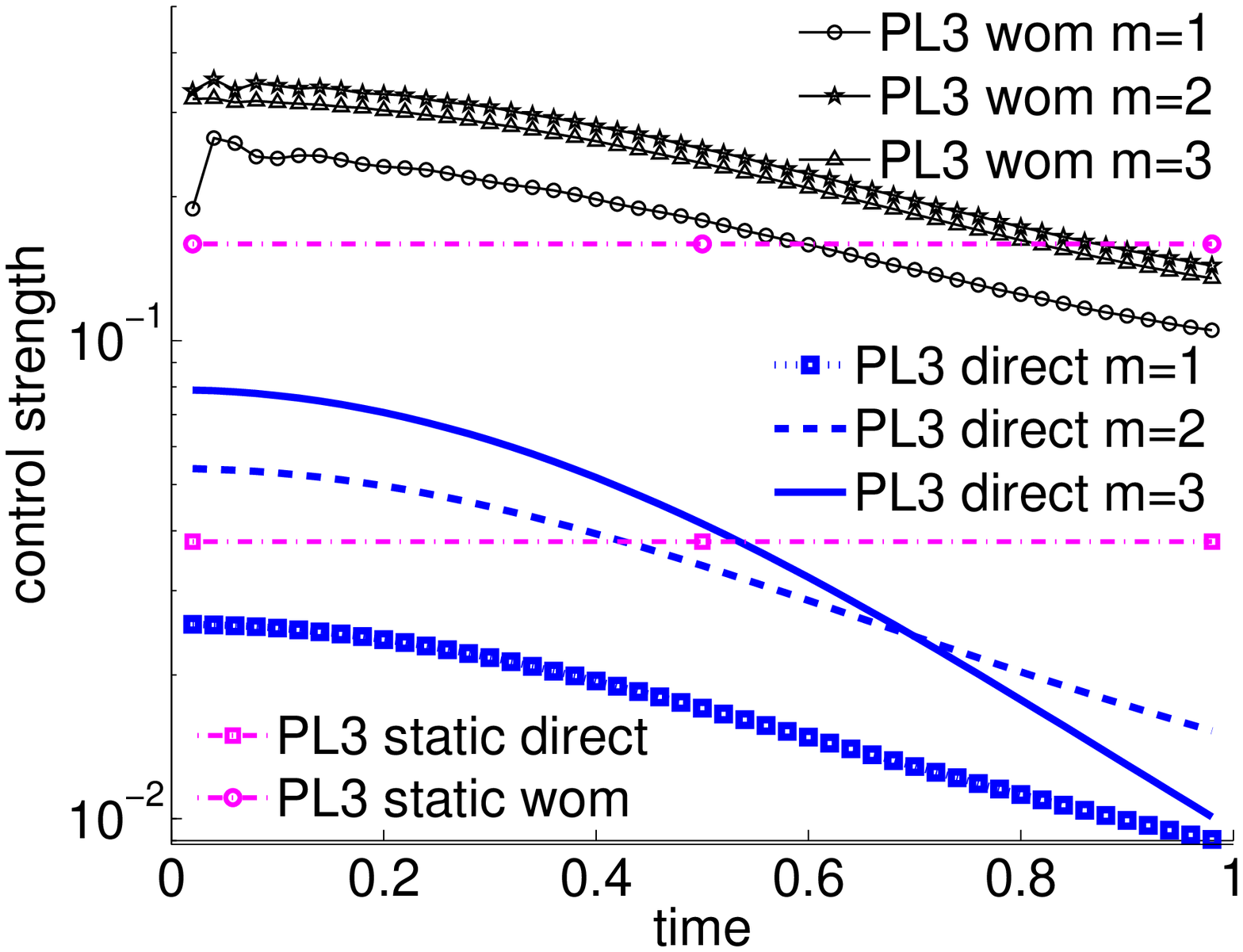} }
\hfill
\subfloat[Resource allocation, PL3, $\uparrow \beta_2(t)$. \label{fig:resource_PL3_increasing_beta}]{
\includegraphics[width=53mm]{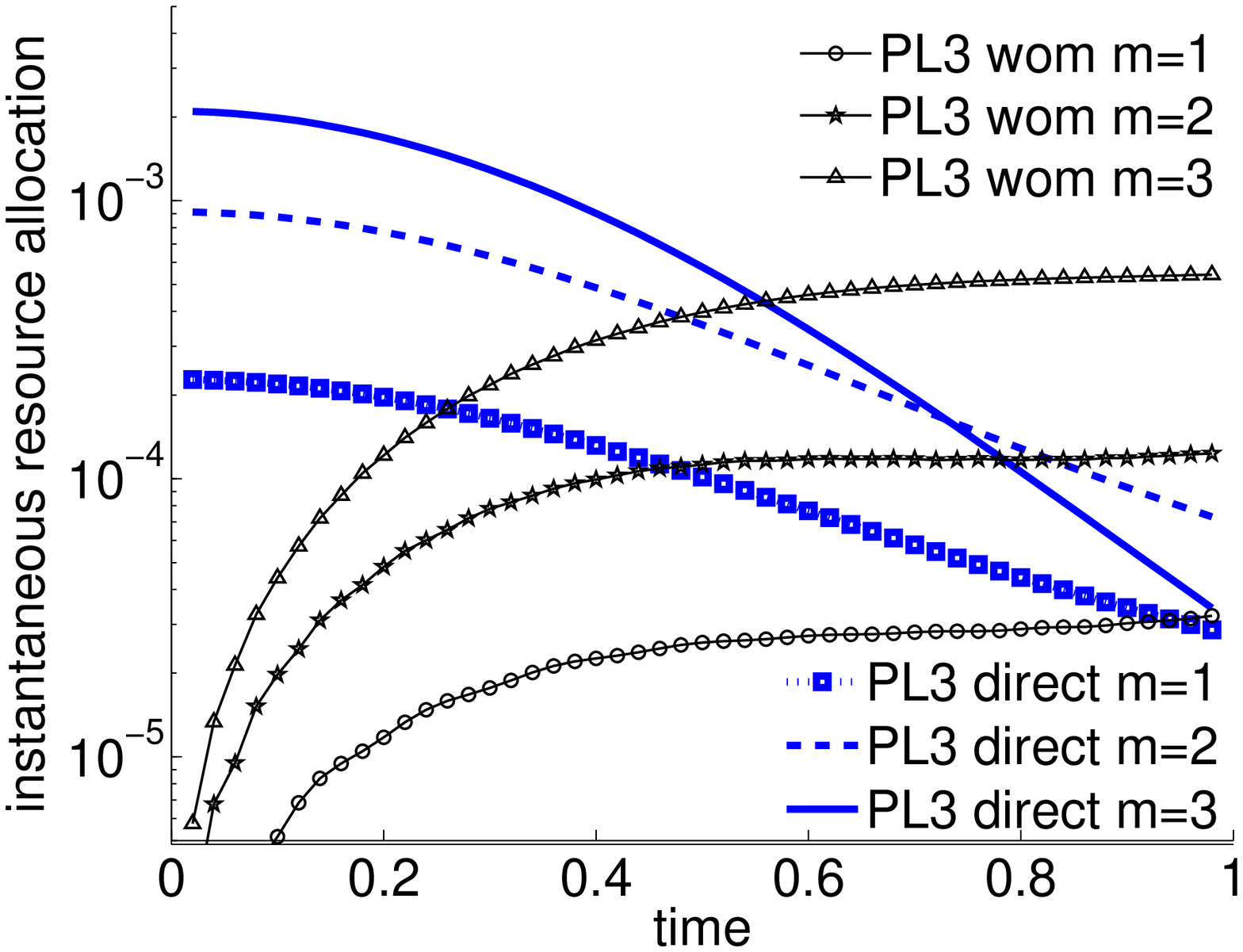} }
\caption{\small{Resource allocation rate for $M=3$ for the case of time varying spreading rate. Parameter values: $T=1$, $i_0=0.01$, $\alpha=0.5$, $u_{max}=0.12$, $v_{max}=0.5$, $\hat b_m=\hat c_m=1~\forall m$, $B = u_{max}^2T/8$, $d = 0.5$.}}
\label{fig:controls_varying_beta}
\end{figure*}

\subsubsection{Time Varying Spreading Rate}

To capture the varying interest level of the population in talking about the subject of the campaign, we have modeled the spreading rate as a time varying quantity. In this section we demonstrate the effect of time-varying spreading rate on the optimal control strategy. We have chosen linearly decreasing and linearly increasing time varying $\beta(t)$ for illustration, defined as:
\begin{eqnarray}
\beta_1(t) =& \beta_M\times (1-t/T), & 0\leq t\leq T, \nonumber \\
\beta_2(t) =& \beta_M\times t/T, & 0\leq t\leq T. \nonumber
\end{eqnarray}
A decreasing $\beta_1(t)$ may be encountered in product promotion campaigns (\emph{e.g.} smartphones), where the interest of the population in the version of the product decreases as it grows old. An increasing $\beta_2(t)$, on the other hand, may be encountered in cases such as poll campaigns, where people are more and more interested in the subject as the polling day approaches.

The controls for the PL3 network for decreasing $\beta_1(t)$ are shown in Fig. \ref{fig:control_PL3_decreasing_beta}, and controls and resource allocation rate curves for increasing $\beta_2(t)$ are shown in Figs. \ref{fig:control_PL3_increasing_beta} and \ref{fig:resource_PL3_increasing_beta} respectively. Here, $\beta_M=0.24$ and degree classes are divided into three groups as in Sec. \ref{sec:results_constant_beta}. The graphs for other cases showed similar trends and are omitted for brevity. Points to note are: 

(i) The word-of-mouth and direct controls are always strong in the beginning stages of the campaign (in-spite of the high costs they incur due to their strength) for both decreasing and increasing spreading rate profiles. More infected nodes at beginning of the campaign leads to more information dissemination in the case of SI process (where there is no recovery like SIR). The word-of-mouth resource allocation (in Fig. \ref{fig:resource_PL3_increasing_beta}) follows the trend, because of slow conversion of susceptibles to infected due to word-of-mouth strategy early on in the campaign (because of low values of $\beta_2(t)$ in early stages). 

(ii) The importance of the degree class groups in spreading the epidemic is the same as in Sec. \ref{sec:results_constant_beta} for ER and PL2/PL3 networks. Thus, the earlier result is reconfirmed, that is, the order of importance for the ER network is: medium degrees $>$ high degrees $>$ low degrees, and that for the PL2/PL3 network is: high degrees $>$ medium degrees $>$ low degrees (omitted graphs showed the same trend).

\subsection {Effect of the Budget}

\begin{figure*}[ht!]
\centering
\subfloat[ER \label{fig:J_vs_B_d_wom_ER}]{
\includegraphics[width=53mm]{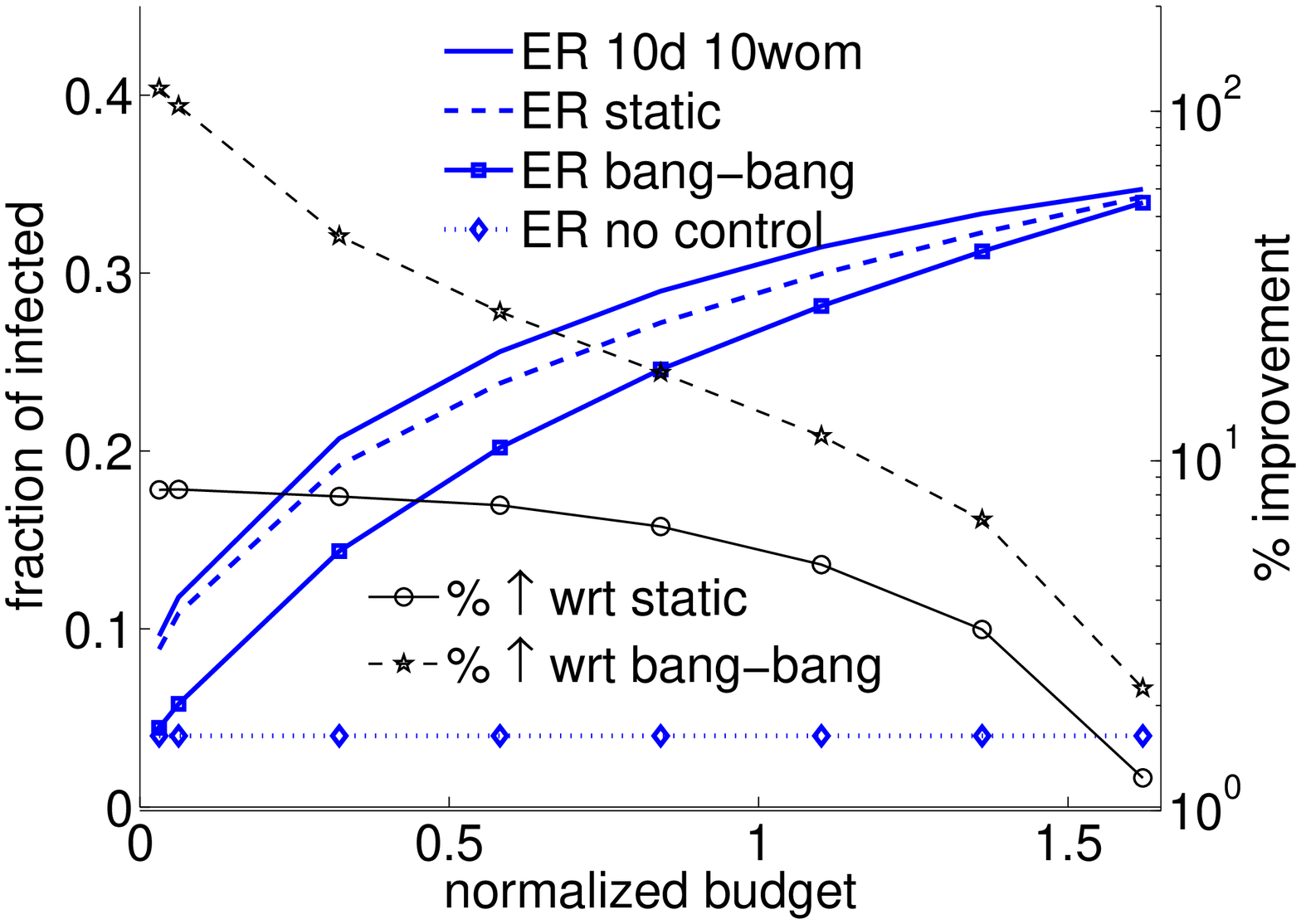} }
\hfill
\subfloat[PL3 \label{fig:J_vs_B_d_wom_PL3}]{
\includegraphics[width=53mm]{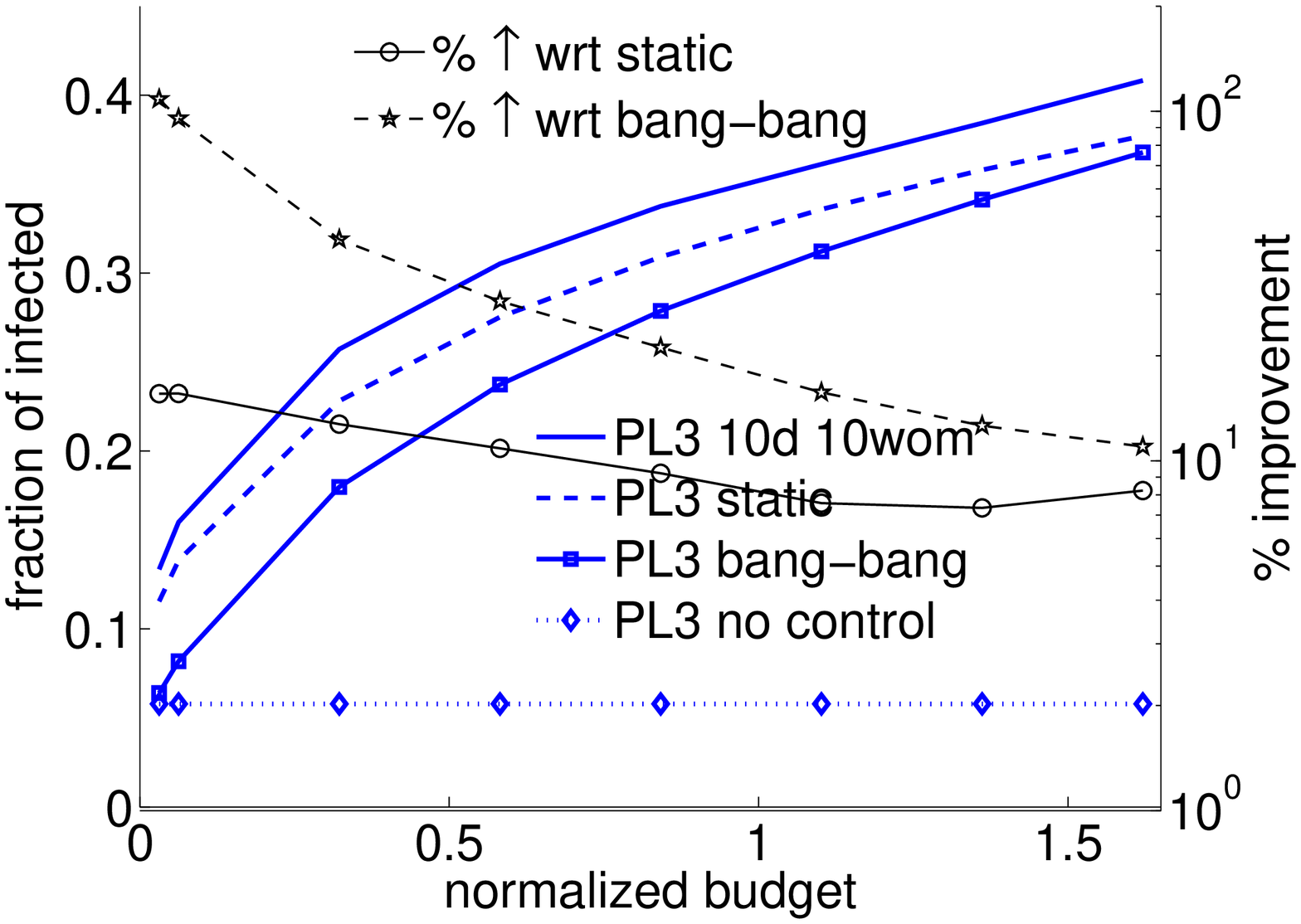} }
\hfill
\subfloat[PL2 \label{fig:J_vs_B_d_wom_PL2}]{
\includegraphics[width=53mm]{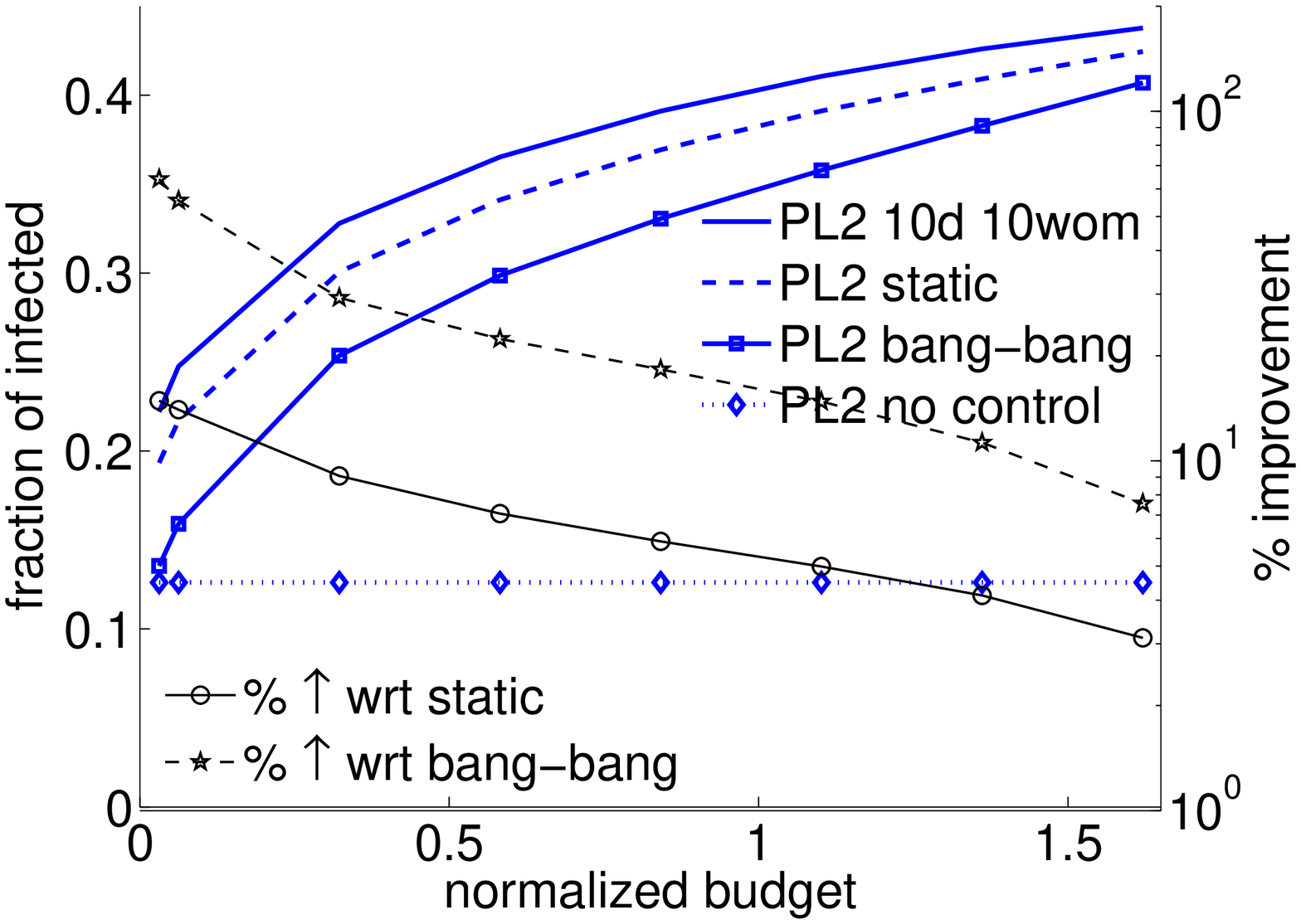} }
\caption{\small{$J$ vs. normalized budget $B$/($u_{max}^2T$), and percentage improvement (right Y-axis) with respect to static and bang-bang control strategies. Parameter values: $T=1$, $i_0=0.01$, $\alpha=0.5$, $\beta=0.12$, $u_{max}=0.12$, $v_{max}=0.5$, $\hat b_m=\hat c_m=1~\forall m$, $d = 0.5$, $M=10$.}}
\label{fig:J_vs_B_d_wom}
\end{figure*}

\begin{figure*}[ht!]
\centering
\subfloat[ER \label{fig:J_vs_B_d_wom_d_ER}]{
\includegraphics[width=53mm]{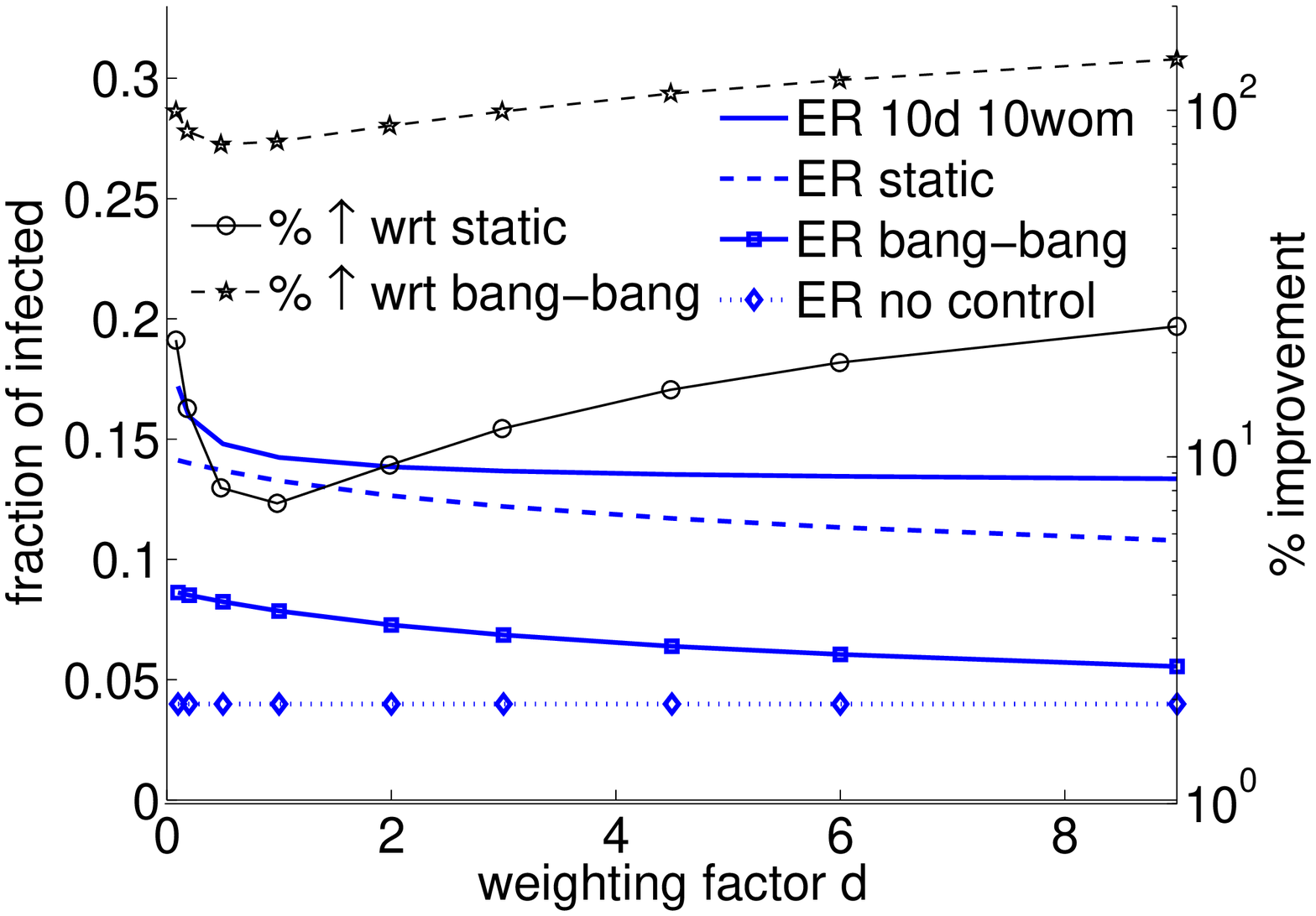} }
\hfill
\subfloat[PL3 \label{fig:J_vs_B_d_wom_d_PL3}]{
\includegraphics[width=53mm]{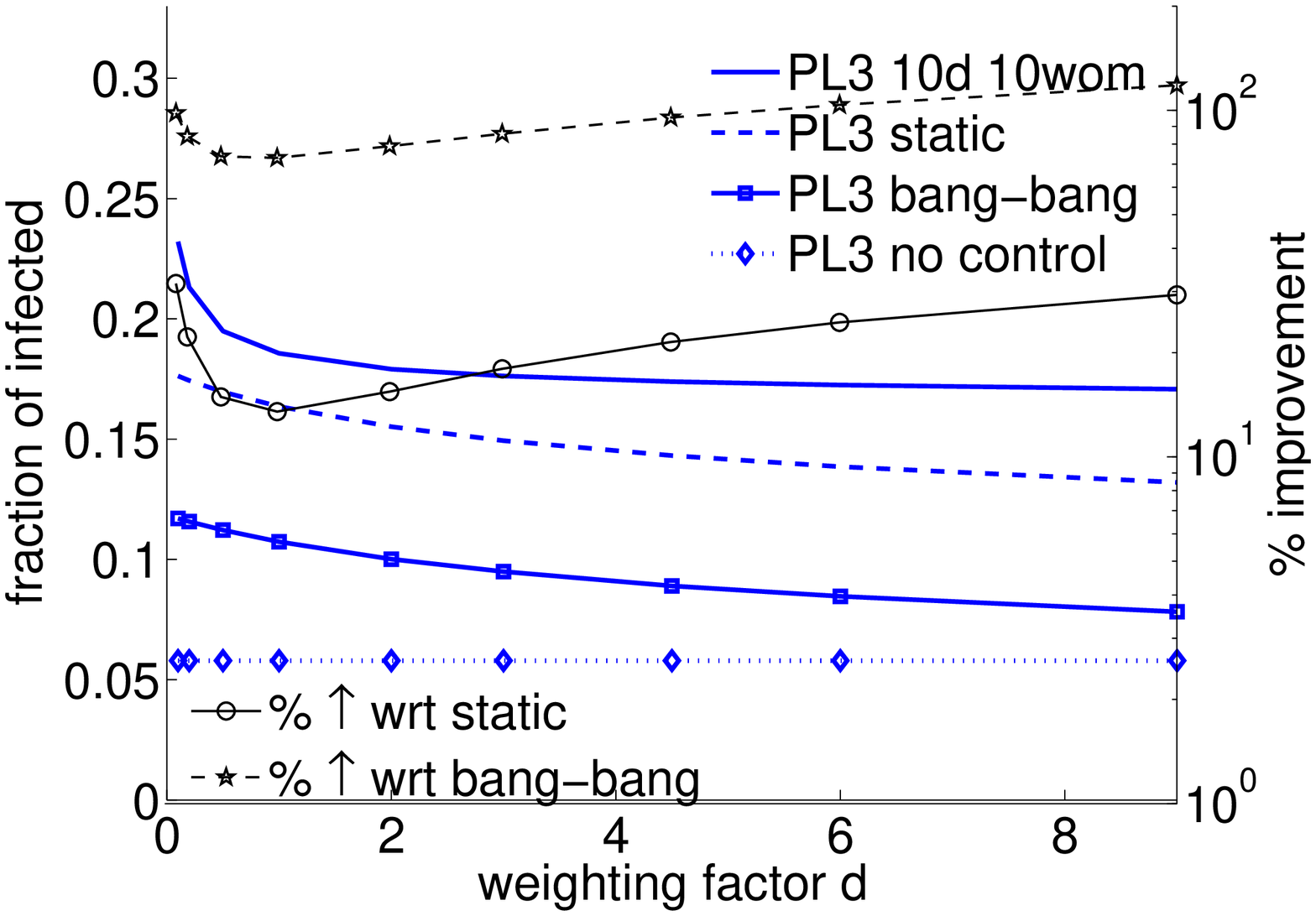} }
\hfill
\subfloat[PL2 \label{fig:J_vs_B_d_wom_d_PL2}]{
\includegraphics[width=53mm]{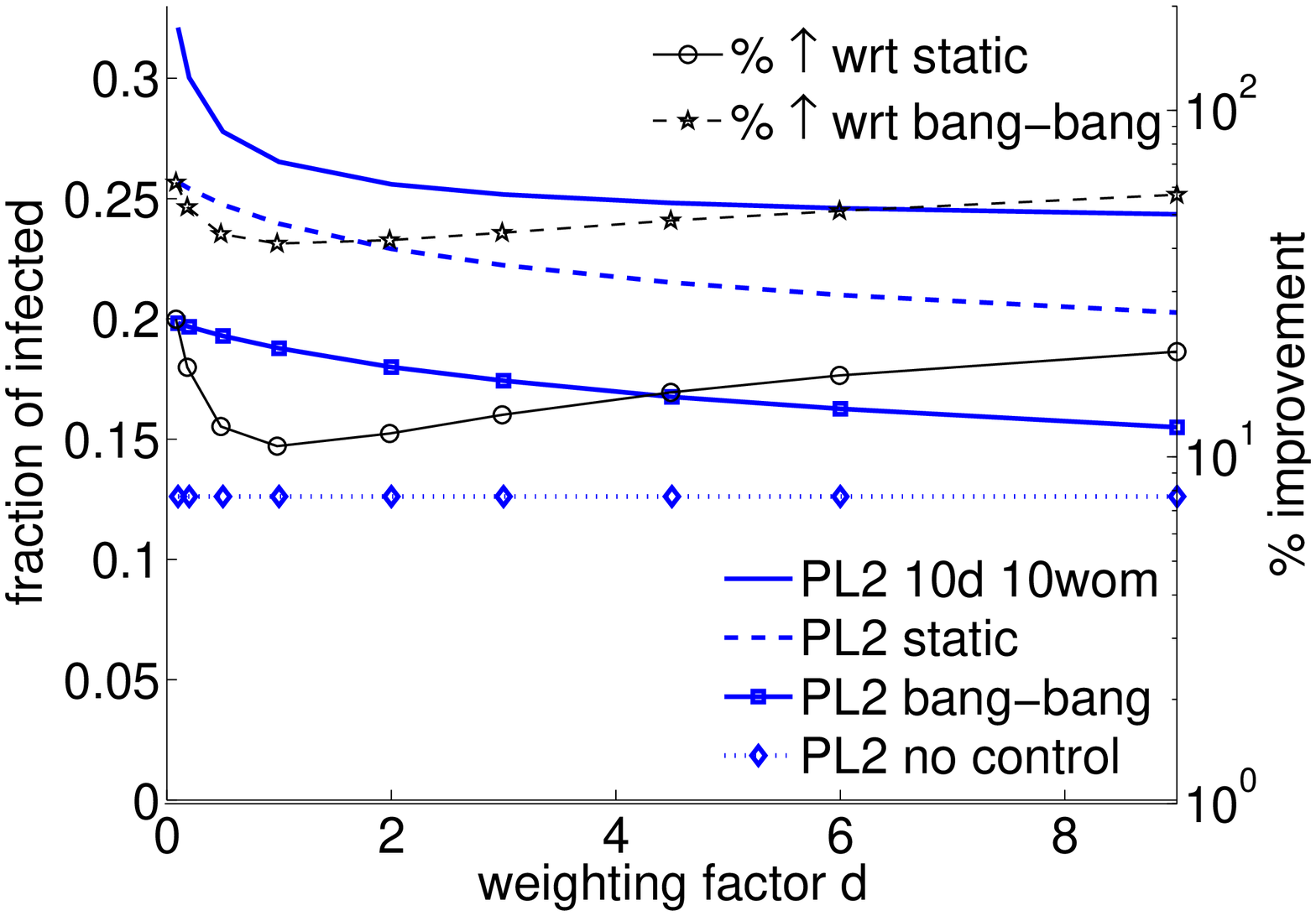} }
\caption{\small{$J$ vs. weighting factor $d$, and percentage improvement (right Y-axis) with respect to static and bang-bang control strategies. Parameter values: $T=1$, $i_0=0.01$, $\alpha=0.5$, $\beta=0.12$, $u_{max}=0.12$, $v_{max}=0.5$, $B=u_{max}^2T/8$, $\hat b_m=\hat c_m=1~\forall m$, $M=10$.}}
\label{fig:J_vs_B_d_wom_d}
\end{figure*}

In the rest of this paper we assume a constant spreading rate. Fig. \ref{fig:J_vs_B_d_wom} studies the effect of the budget on the reward functional $J$, for $M=10$, for all three networks and compares it with static and bang-bang strategies. To divide the degree classes into $M=10$ groups, we use the approach similar to that in Sec. \ref{sec:results_timing_incen_imp_deg}, where degrees were sequentially allotted to the groups, with each group having about one-tenth of the total fraction of the nodes. This approach of creating ten groups remains the same for the rest of the paper. 

From the figures we conclude that (for the parameter values used) the heterogeneous networks (PL2 and PL3) benefit more from the application of the optimal control strategy compared to the ER network, with respect to the static control strategy (concluded from the percentage improvement data in Fig. \ref{fig:J_vs_B_d_wom}). The trend is also roughly true for bang-bang control. The benefit for heterogeneous networks is more because we are still in the regime of low infection and there are untapped susceptible nodes closer to hubs which will benefit from increasing the resource. This does not happen for the case of the ER network. Increase in resource benefits the optimal strategy more than the non-optimal strategies, leading to this trend.

Also, the performance benefits are more for the low and intermediate budgets than for the high budgets. This is expected because at higher budget values, the optimal, static and bang-bang controls are increasingly similar in shape and all of them tends to saturate to the maximum allowed value.

From this figure (and results in the rest of the paper) we see that static strategy outperforms the bang-bang strategy. Although the bang-bang strategy puts the control effort at the most productive stage (beginning of the campaign), the strong intensity incurs more cost (due to the convex cost structure assumed in this section). Compared to the static strategy, the bang-bang strategy is able to exert less total control effort for the same value of budget. It turns out that this trade-off (productive time vs. total effort) is working in favor of the static strategy.

\subsection{Effect of Changing the Relative Cost of Word-of-mouth Controls}

The effect of varying the relative cost of using word-of-mouth control, which is captured by parameter $d$ in the function $c_m(v(t))=d\hat c_mv^2(t)$, is studied in Fig. \ref{fig:J_vs_B_d_wom_d}. Results show that increase in the price of word-of-mouth controls hurts the PL2 network more than the ER network (as suggested by the relative improvement over non-optimal strategies data). The optimal strategy tries to adjust for high costs by allocating more resource to direct controls, so the advantage gained over non-optimal strategies has an increasing trend, but doing so hurts the PL2 network, which uses the word-of-mouth strategy more extensively than the ER network (seen in Sec. \ref{sec:results_timing_incen_imp_deg}). Hence the percentage improvement curve for PL2 is lower than that for ER.

\subsection{Effect of the Spreading Rate}

\begin{figure*}[ht!]
\centering
\subfloat[ER \label{fig:J_vs_beta_d_wom_ER}]{
\includegraphics[width=53mm]{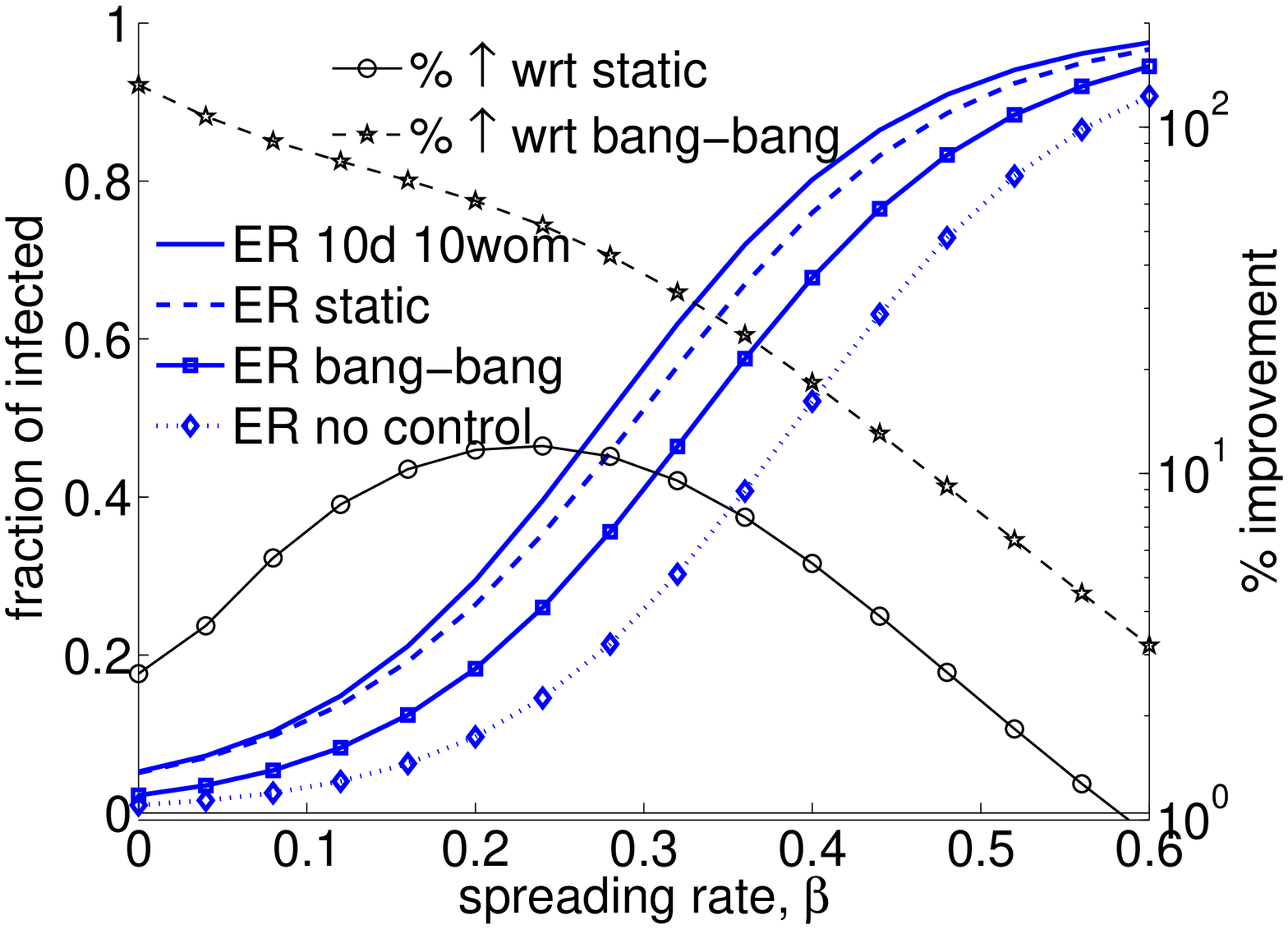} }
\hfill
\subfloat[PL3 \label{fig:J_vs_beta_d_wom_PL3}]{
\includegraphics[width=53mm]{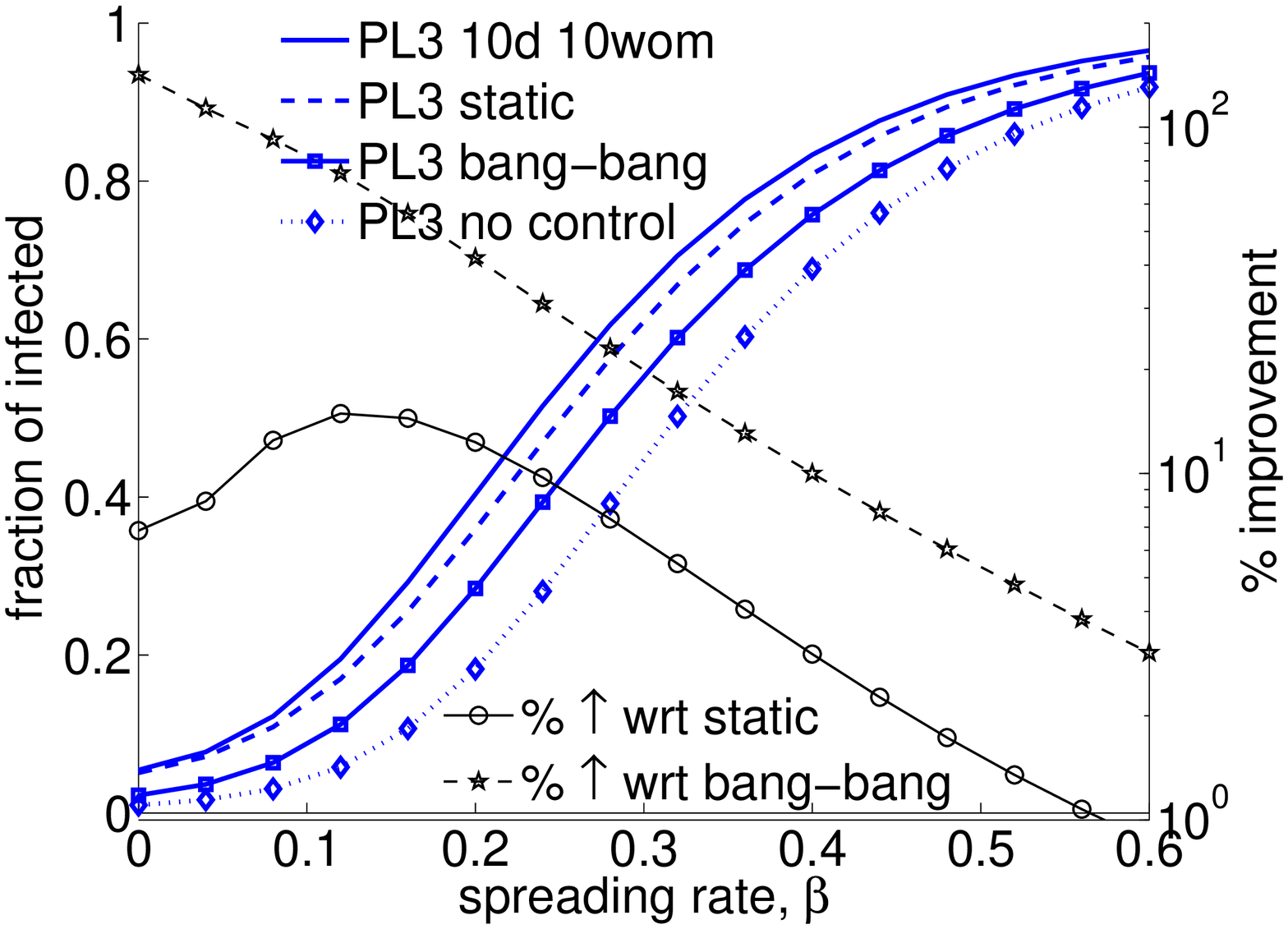} }
\hfill
\subfloat[PL2 \label{fig:J_vs_beta_d_wom_PL2}]{
\includegraphics[width=53mm]{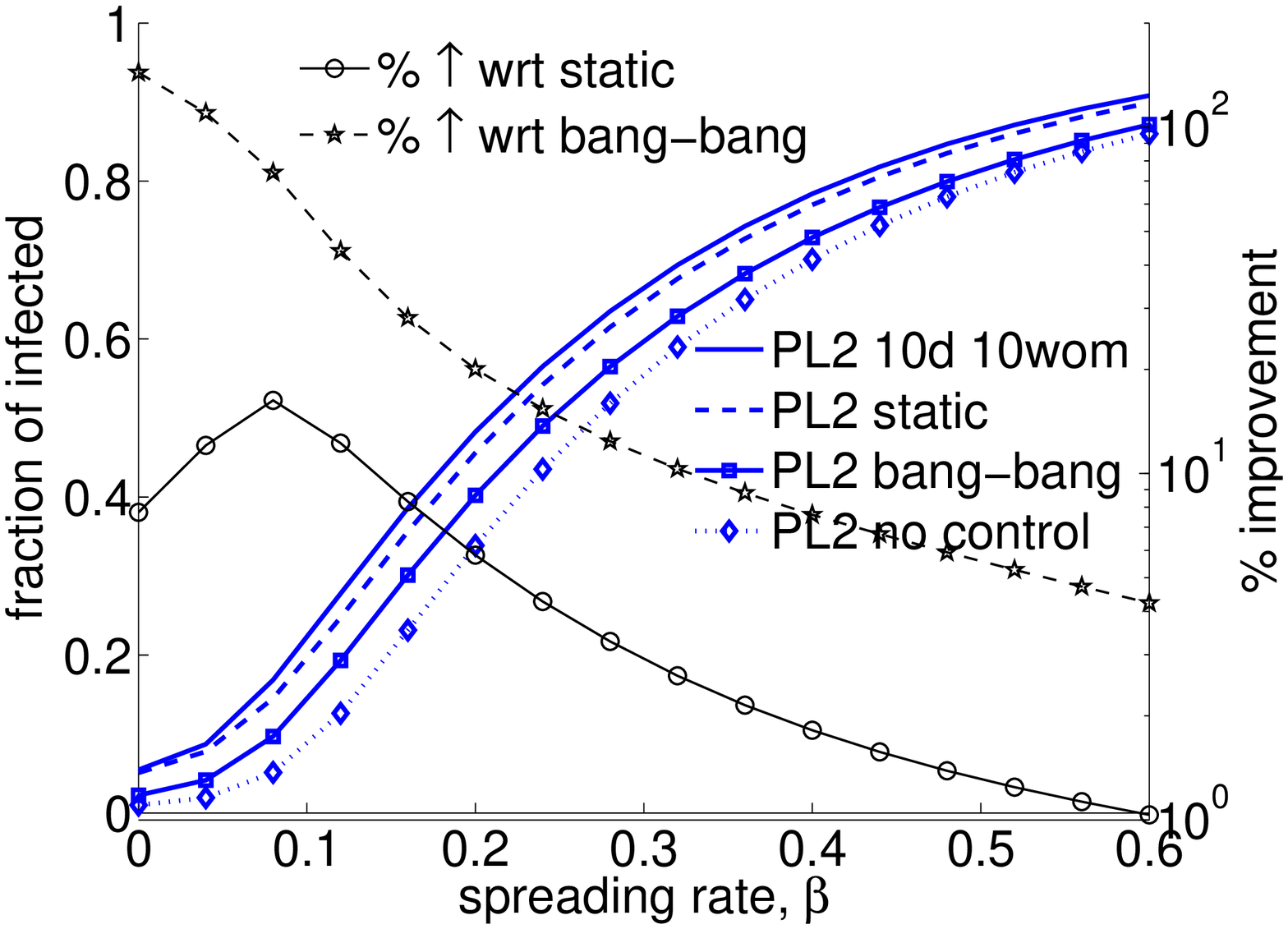} }
\caption{\small{$J$ vs. spreading rate $\beta$, and percentage improvement (right Y-axis) with respect to static and bang-bang control strategies. Parameter values: $T=1$, $i_0=0.01$, $\alpha=0.5$, $u_{max}=0.12$, $v_{max}=0.5$, $B=u_{max}^2T/8$, $\hat b_m=\hat c_m=1~\forall m$, $d = 0.5$, $M=10$.}}
\label{fig:J_vs_beta_d_wom}
\end{figure*}

Fig. \ref{fig:J_vs_beta_d_wom} shows the effect of varying the spreading rate $\beta$ on the reward functional $J$. The optimal control strategy gives notable improvement (say $\geq 10\%$) over the non optimal strategies only for a window of low and intermediate values of $\beta$ (exact window depends on the network used). When the spreading rate is too high, the static or bang-bang strategies will not perform too badly compared to the optimal control strategy.

The figures also reveal that the relative advantage of the optimal and static campaigning strategies over the no campaigning strategy decreases as the spreading rate increases. At higher spreading rates, we are able to reach a greater fraction of the population even without campaigning and hence the importance of campaigning (optimal or static/bang-bang) decreases.

Examining the right part of the three figures carefully, we see that at higher spreading rates, the fraction of population reached through optimal campaigning is more for the ER and PL3 networks than the PL2 network (in-spite of the fact that PL2 has heavier tail, \emph{i.e.}, more hubs with larger degrees). This also happens for the case of no campaigning (or no control case, see also Fig. \ref{fig:sim_vs_theory_SI_config_model}). This is because of the way the networks are chosen. From Fig. \ref{fig:degree_distribution}, PL3 has a minimum degree of 13 and PL2 has minimum degree of 6; about 60\% of the nodes in PL2 have degree of 13 or less (the minimum degree of PL3). We know that an epidemic spreads through the edges, so reaching the low degree nodes, which constitute a significant portion of PL2, is more difficult \emph{at higher levels of infection}. Although some low degree nodes are connected to hubs, but low degree nodes are large in number, and the rest of them which are away from hubs and connected to other low degree nodes, are penetrated slowly. However, this effect is not seen at low levels of infection (low values of $J$) because we mostly infect the nodes close to the hubs.

\subsection{Effect of the Cost of Influencing Various Target Groups}

\begin{figure*}[ht!]
\centering
\subfloat[ER \label{fig:J_vs_small_b_d_wom_ER}]{
\includegraphics[width=53mm]{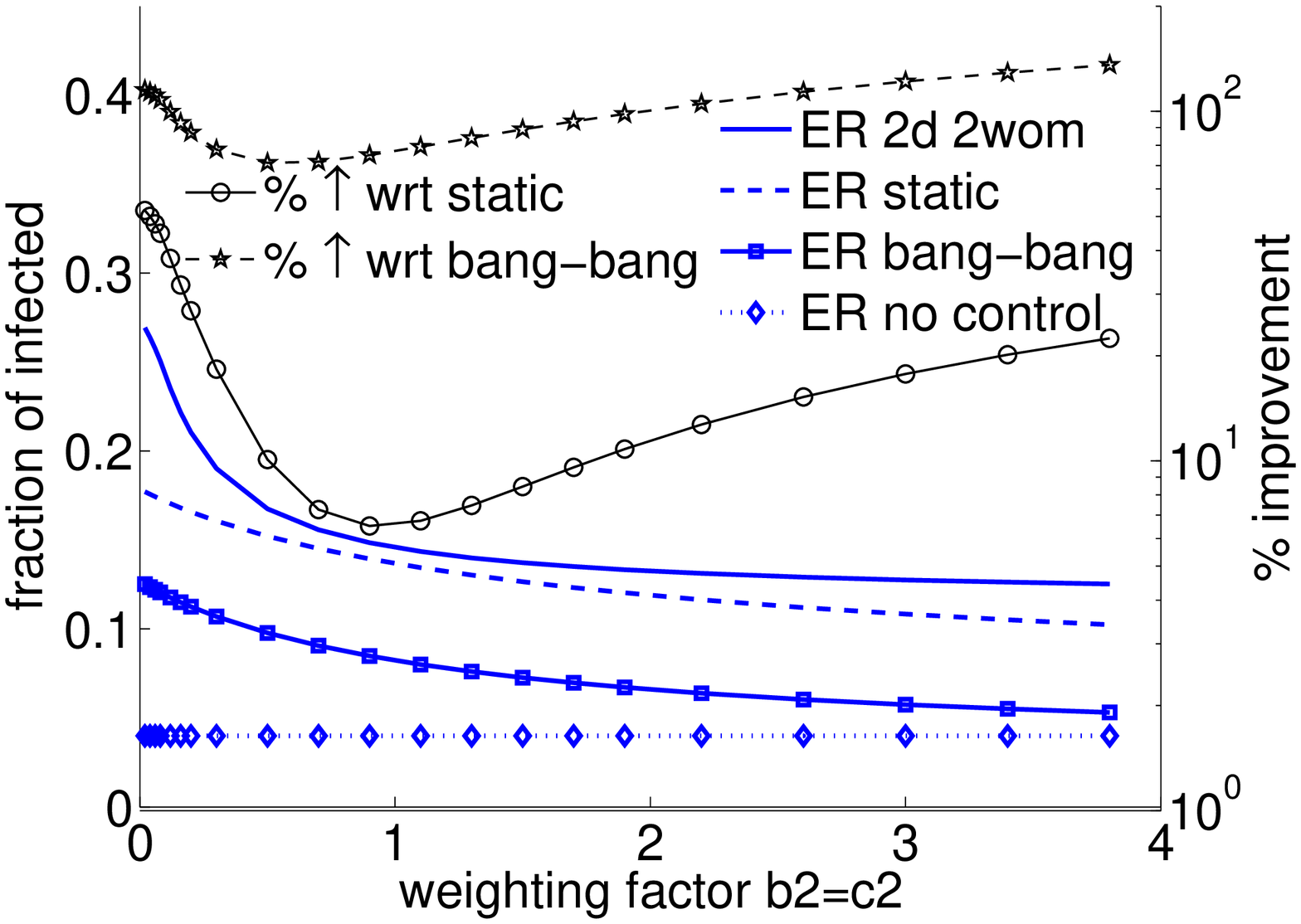} }
\hfill
\subfloat[PL3 \label{fig:J_vs_small_b_d_wom_PL3}]{
\includegraphics[width=53mm]{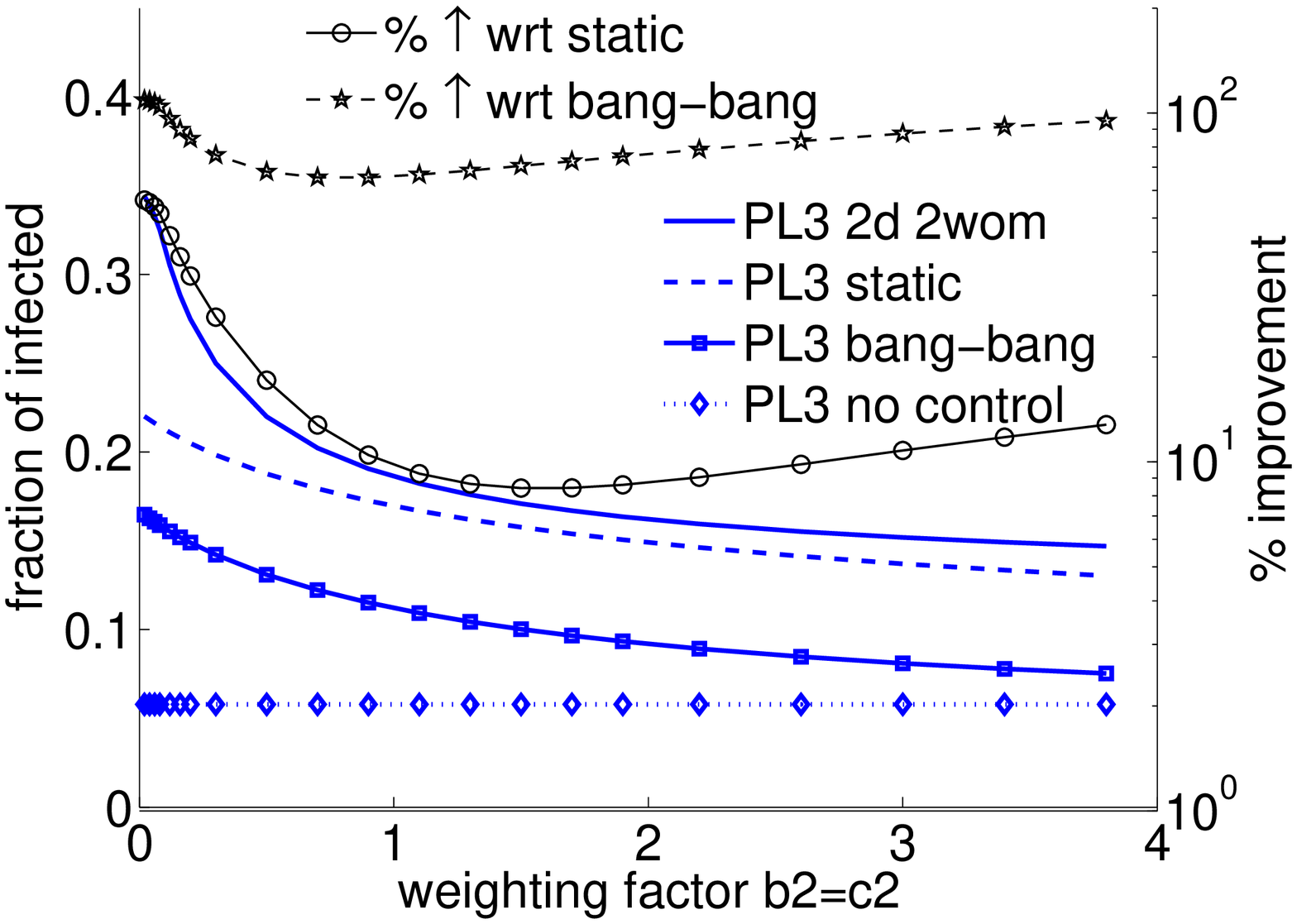} }
\hfill
\subfloat[PL2 \label{fig:J_vs_small_b_d_wom_PL2}]{
\includegraphics[width=53mm]{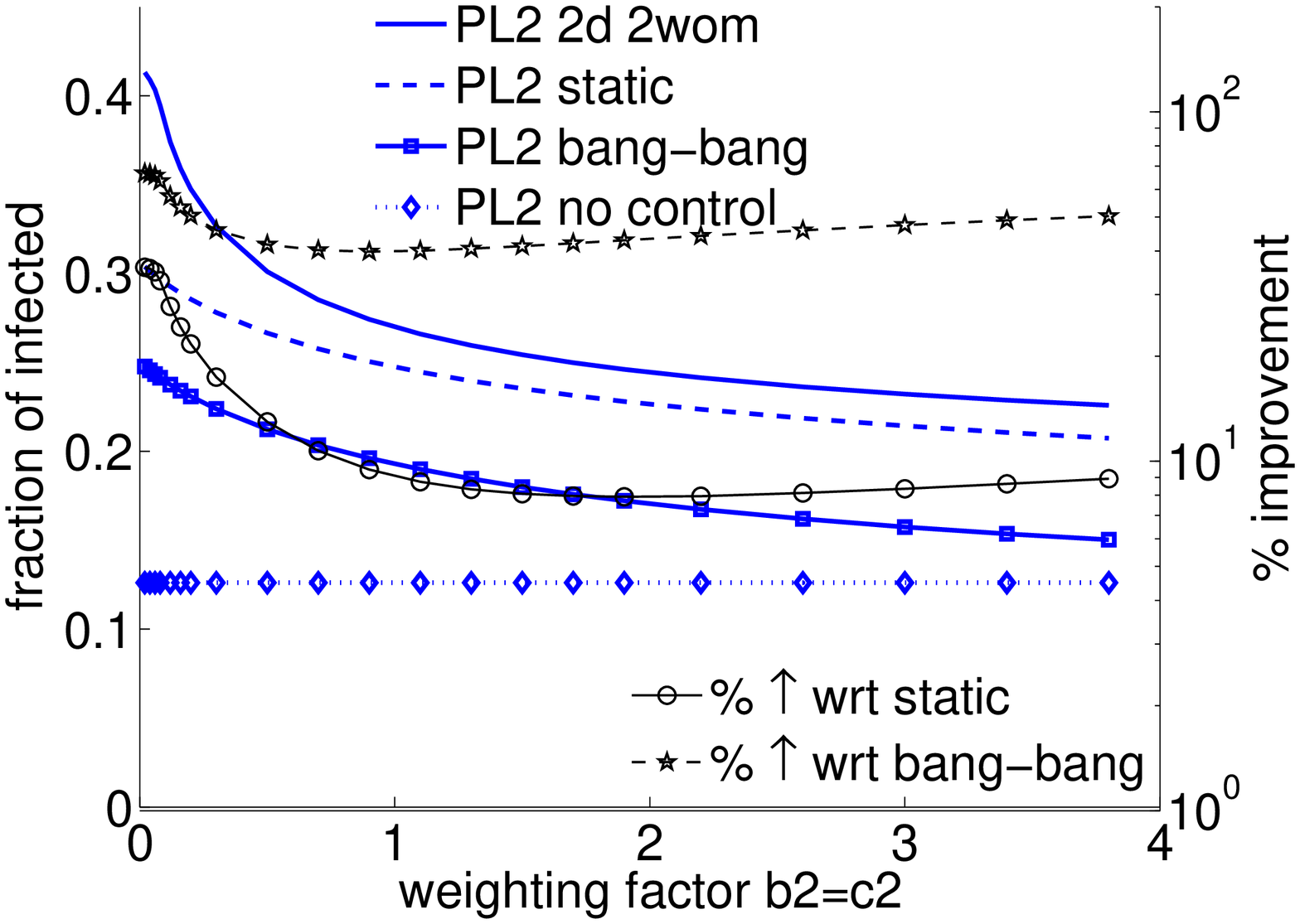} }
\caption{\small{$J$ vs. weighting factor $\hat b_2=\hat c_2$, and percentage improvement (right Y-axis) with respect to static and bang-bang control strategies. Parameter values: $M=2$, $T=1$, $i_0=0.01$, $\alpha=0.5$, $\beta=0.12$, $u_{max}=0.12$, $v_{max}=0.5$, $B=u_{max}^2T/8$, $\hat b_1=\hat c_1=1$, $d = 0.5$.}}
\label{fig:J_vs_small_b_d_wom}
\end{figure*}

We study the effect of changing the cost of influencing and incentivizing the target groups, for the case when the degree classes are divided into $M=2$ groups. We change this cost by changing the weighting parameters $\hat b_m$ and $\hat c_m$ in the functions $b_m()$ and $c_m()$ defined in Sec. \ref{sec:defalut_parameters}. The results are shown in Fig. \ref{fig:J_vs_small_b_d_wom}. The weighting parameters for the group containing lower degree classes, $\hat b_1$ and $\hat c_1$ are fixed to 1, and $\hat b_2=\hat c_2$ are varied. As the skewness in the costs increases, the advantage obtained by the optimal control strategy over the non-optimal strategies also increases. This is so because the optimal strategy starts allocating more resource to the cheaper group. Since high degree classes are more important for spreading in the case of PL2 and PL3 networks, the relative advantage of optimal strategies over non-optimal strategies in these networks, at high values of $\hat b_2=\hat c_2$ is less, in comparison to the ER network (where degrees of the nodes present in the network are distributed more homogeneously).

\subsection{Effect of the Number of Seeds at the Start of the Epidemic}

\begin{figure*}[ht!]
\centering
\subfloat[ER \label{fig:J_vs_i0_d_wom_ER}]{
\includegraphics[width=53mm]{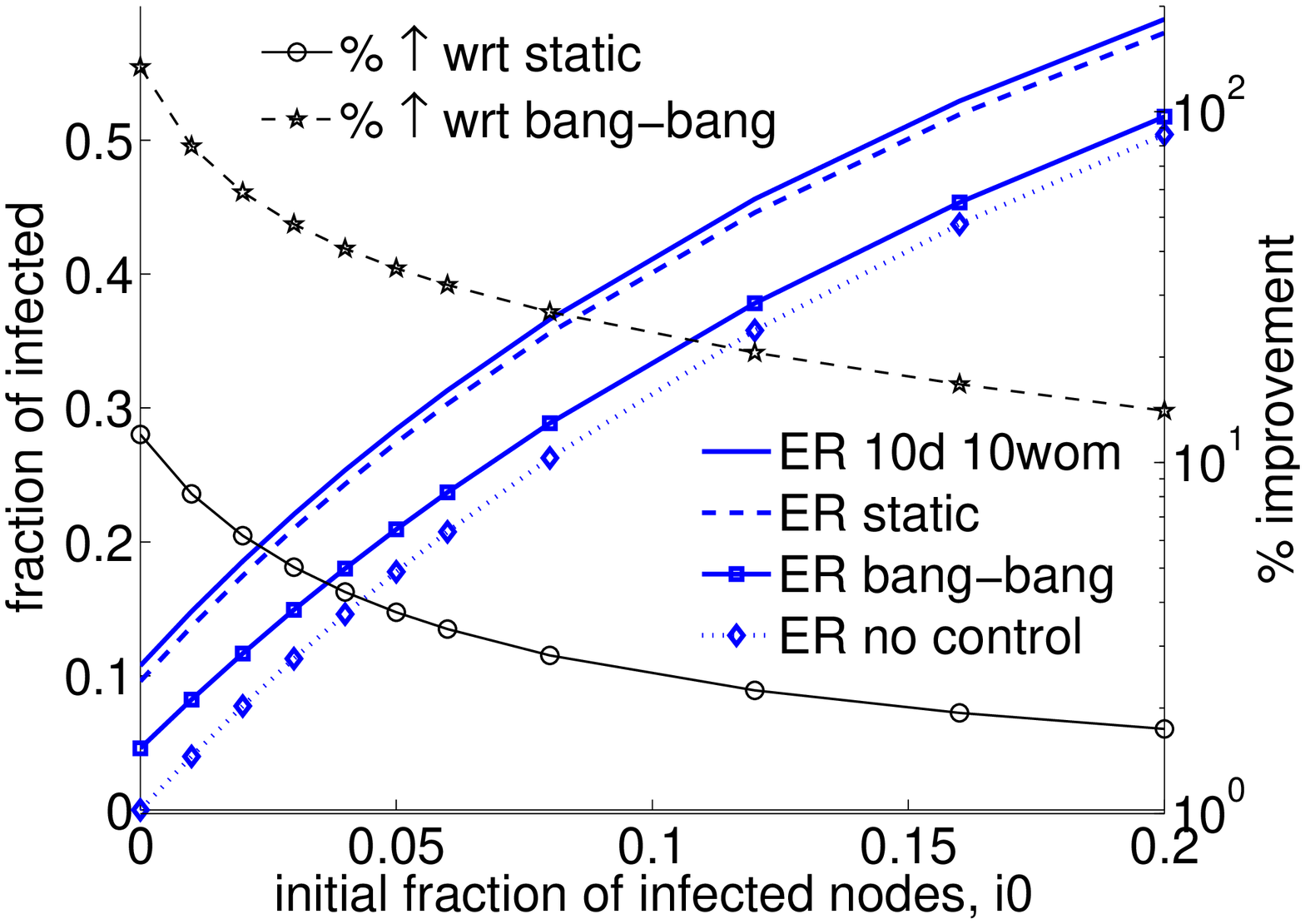} }
\hfill
\subfloat[PL3 \label{fig:J_vs_i0_d_wom_PL3}]{
\includegraphics[width=53mm]{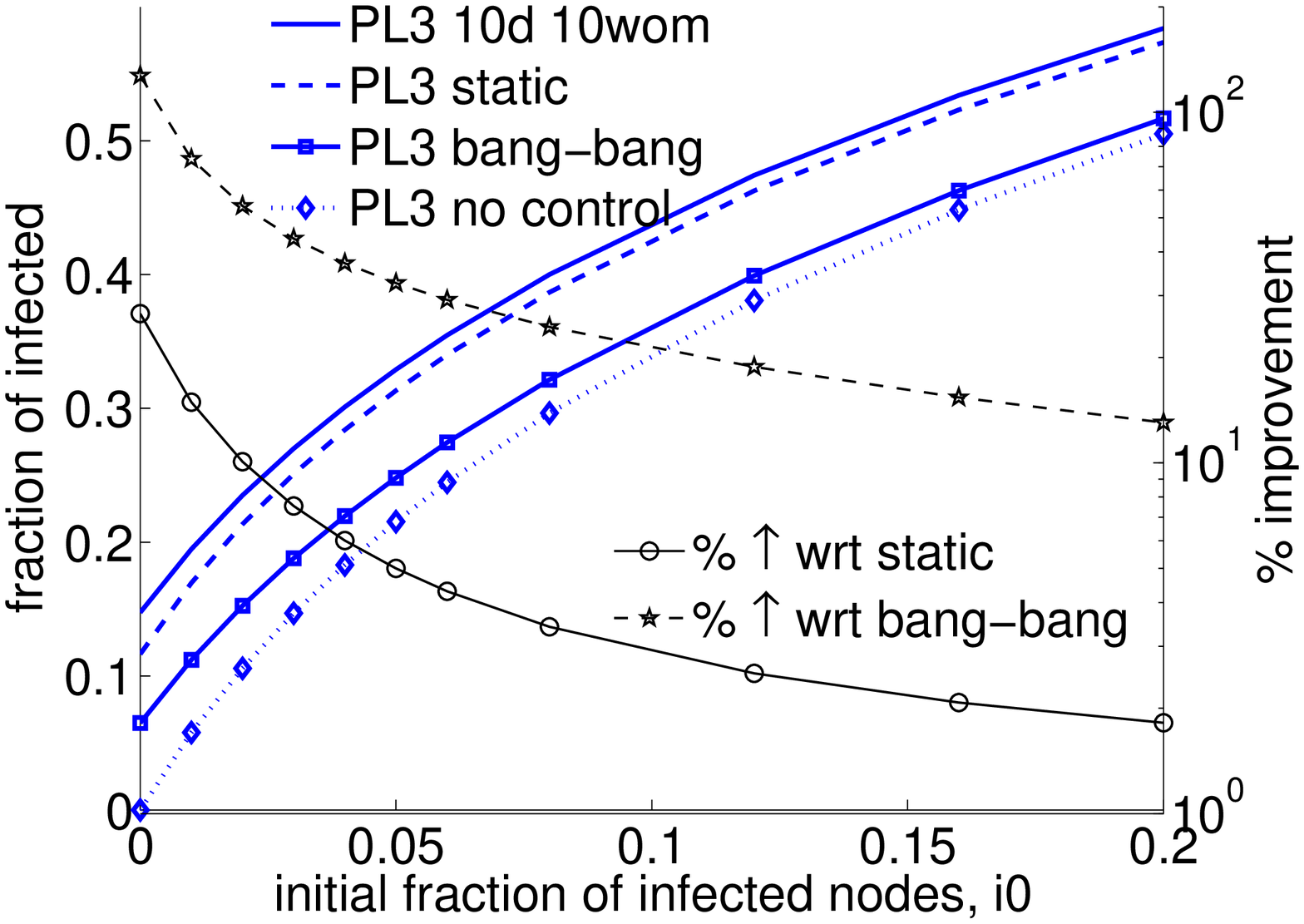} }
\hfill
\subfloat[PL2 \label{fig:J_vs_i0_d_wom_PL2}]{
\includegraphics[width=53mm]{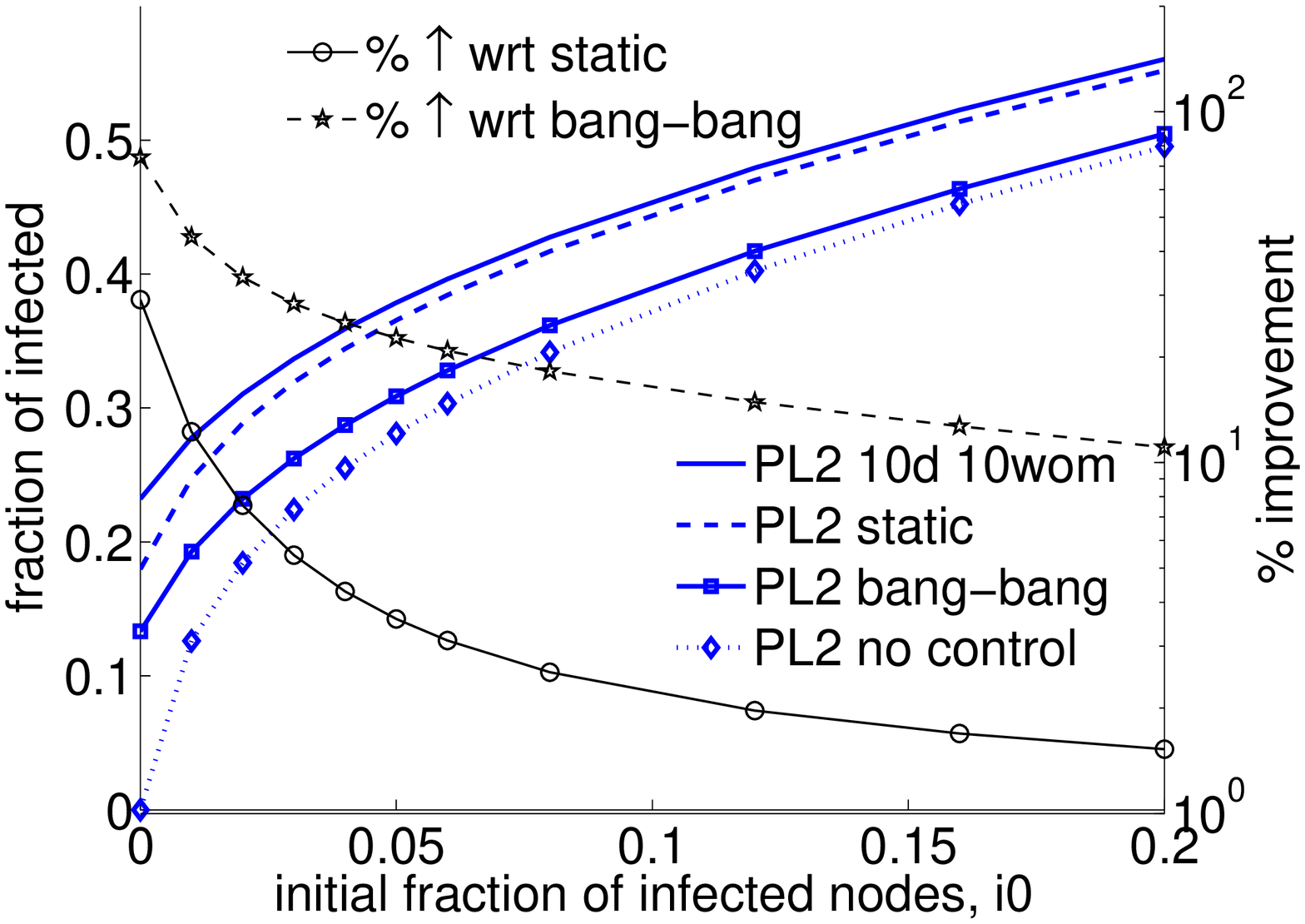} }
\caption{\small{$J$ vs. $i_0$ (seeds), and percentage improvement (right Y-axis) with respect to static and bang-bang control strategies. Parameter values: $T=1$, $\alpha=0.5$, $\beta=0.12$, $u_{max}=0.12$, $v_{max}=0.5$, $B=u_{max}^2T/8$, $\hat b_m=\hat c_m=1~\forall m$, $d = 0.5$, $M=10$.}}
\label{fig:J_vs_i0_d_wom}
\end{figure*}

Fig. \ref{fig:J_vs_i0_d_wom} studies the effect of the fraction of population already infected at the start of the epidemic, $i_0$ (seeds) on $J$. As $i_0$ increases, the improvement offered by the optimal strategy over the non-optimal strategies decreases quickly. Thus, one should only bother about calculating and implementing the optimal strategy if starting the campaign early. However, it should be noted that for most real world applications, this condition is satisfied most of the time; we become aware of some information as a result of the ongoing campaign.

\section{Limitations and Future Work}

We list some of the limitations and potential future directions in which the work presented in this paper can be extended:

(1) We have only used information from the degree distribution of the network to prescribe the optimal allocation of the resource. The effects of some of the other network properties like degree-degree correlations, community structure, clustering etc. on the optimal strategy still need to be explored. The degree based compartmental model may not be accurate for some of these situations, and more complicated methods, such as pair approximation methods \cite[Sec. 17.10.1]{newman2010networks}, 
may be required. However, pair approximation methods lead to many more differential equations compared to the degree based compartmental models. This will lead to significant computational challenges because the system is solved multiple times to compute (optimal) controls.

(2) Degree is only one measure of node centrality in the network. Exploring the effect of other centrality measures on campaigning resource allocation forms another interesting direction of future work.

\section{Conclusion}

We studied the problem of optimal resource allocation over time, for maximizing the spread of information among nodes (individuals) connected in a network, under fixed budget constraints. We used the degree based compartmental model for the Susceptible-Infected epidemics to model information propagation on networks. Nodes with the same degree are aggregated into the same degree class. The degree classes are divided into $M$ groups, each influenced by a direct and a word-of-mouth control, the two strategies to accelerate information spreading.

Our formulation leads to an optimal control problem the solution to which provides the best allocation of the resource over (i) the campaign duration, (ii) strategies, and (iii) the $M$ groups, to maximize the information spread. In the process, the solution also identifies the important groups which should be allotted more resources to maximize the epidemic size. The order of importance of the groups depends on the network.

We proved the existence of a solution to the optimal control problem with non-linear isoperimetric constraints using novel techniques. The optimal control problem was solved by converting it into a non-linear static optimization problem. Our formulation works for arbitrary degree distributions which allowed us to study the effects of network topology on epidemic control. We presented results for Erd\H os-R\'enyi and scale free networks (with power law exponents 2 and 3) and studied the effects of various model parameters on the optimal control strategy. We have quantified the improvement offered by the optimal control strategy over the static and bang-bang control strategies.

Our results show that for certain scenarios, the performance gains achieved by the optimal strategy are only marginal compared to the non-optimal strategies. Controls are stronger in early stages of the campaign. For the case when degree classes are divided into three groups, for scale free (heterogeneous) networks, the order of importance of the groups in message spreading is: (the importance is based on the resource allocated by the optimal strategy to the groups) high degrees $>$ medium degrees $>$ low degrees. In contrast, for the Erd\H os-R\'enyi (homogeneous) network, it is: medium degrees $>$ high degrees $>$ low degrees. Word-of-mouth strategy is more important for maximizing information diffusion in heterogeneous networks than homogeneous networks. Optimal campaigning leads to significant performance gains (over a static strategy) only if campaigning starts early, for topics which have low to intermediate spreading rates and for low to intermediate budget constraints.

The results presented in this paper may be of interest to (election and crowdfunding) campaign managers and product marketing managers working to disseminate a piece of information on a social network. The framework presented in this paper, with some modifications, can also be used for biological epidemic mitigation on networks.

\appendices
\section{Neighbor Degree Distribution}
\label{appendix:neighbor_degree_distribution}
Here we argue that in the configuration model, with degree distribution $p_k,~k\in \mathbb K$, neighbor degree distribution is given by, $r_k=kp_k/\bar k$, where $\bar k$ is mean degree of the network \cite[Sec. 13.3]{newman2010networks}.

When an edge is cut, it leads to two half edges, one adjoined to each node which the original edge is joined to. The total number of half edges in the network is given by $2m$ an even positive integer. In the configuration model, a half edge originating from any node is equally likely to terminate at any other node. If total number of nodes in the network is $N$, total number of nodes of degree $k$ is $p_kN$. The total number of half edges corresponding to this is $kp_kN$. Thus, the probability that one of them will be paired with the tagged half edge originating from the tagged node is $kp_kN/(2m-1)$ (we leave out the edge from the tagged node itself). In the limit of large number of nodes, this is $\approx kp_kN/2m$. Notice that $2m/N=\bar k$. Thus the probability that the tagged edge will terminate on a neighbor that has degree $k$ is $kp_k/\bar k$.


\ifCLASSOPTIONcaptionsoff
  \newpage
\fi


\small
\bibliographystyle{IEEEtran}
\bibliography{bibliography_database}

\begin{thebibliography}{10}
\providecommand{\url}[1]{#1}
\csname url@samestyle\endcsname
\providecommand{\newblock}{\relax}
\providecommand{\bibinfo}[2]{#2}
\providecommand{\BIBentrySTDinterwordspacing}{\spaceskip=0pt\relax}
\providecommand{\BIBentryALTinterwordstretchfactor}{4}
\providecommand{\BIBentryALTinterwordspacing}{\spaceskip=\fontdimen2\font plus
\BIBentryALTinterwordstretchfactor\fontdimen3\font minus
  \fontdimen4\font\relax}
\providecommand{\BIBforeignlanguage}[2]{{%
\expandafter\ifx\csname l@#1\endcsname\relax
\typeout{** WARNING: IEEEtran.bst: No hyphenation pattern has been}%
\typeout{** loaded for the language `#1'. Using the pattern for}%
\typeout{** the default language instead.}%
\else
\language=\csname l@#1\endcsname
\fi
#2}}
\providecommand{\BIBdecl}{\relax}
\BIBdecl

\bibitem{goldenberg2009role}
J.~Goldenberg, S.~Han, D.~R. Lehmann, and J.~W. Hong, ``The role of hubs in the
  adoption process,'' \emph{Journal of Marketing}, vol.~73, no.~2, pp. 1--13,
  2009.

\bibitem{karsai2011small}
M.~Karsai, M.~Kivel{\"a}, R.~K. Pan, K.~Kaski, J.~Kert{\'e}sz, A.-L.
  Barab{\'a}si, and J.~Saram{\"a}ki, ``Small but slow world: How network
  topology and burstiness slow down spreading,'' \emph{Physical Review E},
  vol.~83, no.~2, p. 025102, 2011.

\bibitem{belleflamme2013crowdfunding}
P.~Belleflamme, T.~Lambert, and A.~Schwienbacher, ``Crowdfunding: Tapping the
  right crowd,'' \emph{Journal of Business Venturing}, 2013.

\bibitem{shaw2010enhanced}
L.~B. Shaw and I.~B. Schwartz, ``Enhanced vaccine control of epidemics in
  adaptive networks,'' \emph{Physical Review E}, vol.~81, no.~4, p. 046120,
  2010.

\bibitem{dykman2008disease}
M.~I. Dykman, I.~B. Schwartz, and A.~S. Landsman, ``Disease extinction in the
  presence of random vaccination,'' \emph{Physical Review Letters}, vol. 101,
  no.~7, p. 078101, 2008.

\bibitem{asano2008optimal}
E.~Asano, L.~J. Gross, S.~Lenhart, and L.~A. Real, ``Optimal control of vaccine
  distribution in a rabies metapopulation model,'' \emph{Mathematical
  Biosciences and Engineering}, vol.~5, no.~2, pp. 219--238, 2008.

\bibitem{gaff2009optimal}
H.~Gaff and E.~Schaefer, ``Optimal control applied to vaccination and treatment
  strategies for various epidemiological models,'' \emph{Mathematical
  Biosciences and Engineering}, vol.~6, no.~3, pp. 469--492, 2009.

\bibitem{yan2008optimal}
X.~Yan and Y.~Zou, ``Optimal internet worm treatment strategy based on the
  two-factor model,'' \emph{ETRI Journal}, vol.~30, no.~1, pp. 81--88, 2008.

\bibitem{zhu2012optimal}
Q.~Zhu, X.~Yang, L.-X. Yang, and C.~Zhang, ``Optimal control of computer virus
  under a delayed model,'' \emph{Applied Mathematics and Computation}, vol.
  218, no.~23, pp. 11\,613--11\,619, 2012.

\bibitem{dayama2012optimal}
P.~Dayama, A.~Karnik, and Y.~Narahari, ``Optimal incentive timing strategies
  for product marketing on social networks,'' in \emph{International Conference
  on Autonomous Agents and Multiagent Systems}, 2012, pp. 703--710.

\bibitem{youssef2013mitigation}
M.~Youssef and C.~Scoglio, ``Mitigation of epidemics in contact networks
  through optimal contact adaptation,'' \emph{Mathematical Biosciences and
  Engineering}, vol.~10, no.~4, pp. 1227--1251, 2013.

\bibitem{chierichetti2009rumor}
F.~Chierichetti, S.~Lattanzi, and A.~Panconesi, ``Rumor spreading in social
  networks,'' in \emph{Automata, Languages and Programming}.\hskip 1em plus
  0.5em minus 0.4em\relax Springer, 2009, pp. 375--386.

\bibitem{pittel1987spreading}
B.~Pittel, ``On spreading a rumor,'' \emph{SIAM Journal on Applied
  Mathematics}, vol.~47, no.~1, pp. 213--223, 1987.

\bibitem{karnik2012optimal}
A.~Karnik and P.~Dayama, ``Optimal control of information epidemics,'' in
  \emph{International Conference on Communication Systems and Networks}.\hskip
  1em plus 0.5em minus 0.4em\relax IEEE, 2012, pp. 1--7.

\bibitem{kandhway2014run}
K.~Kandhway and J.~Kuri, ``{How to run a campaign: Optimal control of SIS and
  SIR information epidemics},'' \emph{Applied Mathematics and Computation},
  vol. 231, pp. 79--92, 2014.

\bibitem{kandhway2014optimal}
------, ``{Optimal control of information epidemics modeled as Maki Thompson
  rumors},'' \emph{Communications in Nonlinear Science and Numerical
  Simulation}, vol.~19, pp. 4135--4147, 2014.

\bibitem{khouzani2011optimal}
M.~H.~R. Khouzani, S.~Sarkar, and E.~Altman, ``Optimal control of epidemic
  evolution,'' in \emph{International Conference on Computer
  Communications}.\hskip 1em plus 0.5em minus 0.4em\relax IEEE, 2011, pp.
  1683--1691.

\bibitem{belen2008behaviour}
S.~Belen, ``The behaviour of stochastic rumours,'' \emph{Ph.D. Dissertation,
  University of Adelaide}, 2008.

\bibitem{sethi2008optimal}
S.~P. Sethi, A.~Prasad, and X.~He, ``Optimal advertising and pricing in a
  new-product adoption model,'' \emph{Journal of Optimization Theory and
  Applications}, vol. 139, no.~2, pp. 351--360, 2008.

\bibitem{fleming1975deterministic}
W.~H. Fleming and R.~W. Rishel, \emph{Deterministic and Stochastic Optimal
  Control}.\hskip 1em plus 0.5em minus 0.4em\relax Springer-Verlag, 1975.

\bibitem{barrat2008dynamical}
A.~Barrat, M.~Barth\'el\'emy, and A.~Vespignani, \emph{Dynamical Processes on
  Complex Networks}.\hskip 1em plus 0.5em minus 0.4em\relax Cambridge
  University Press, 2008.

\bibitem{newman2010networks}
M.~Newman, \emph{Networks: An Introduction}.\hskip 1em plus 0.5em minus
  0.4em\relax Oxford University Press, 2010.

\bibitem{atkinson2009theoretical}
K.~Atkinson and W.~Han, \emph{Theoretical Numerical Analysis: a Functional
  Analysis Framework}.\hskip 1em plus 0.5em minus 0.4em\relax Springer, 2009.

\bibitem{rudin1964principles}
W.~Rudin, \emph{Principles of Mathematical Analysis}.\hskip 1em plus 0.5em
  minus 0.4em\relax McGraw-Hill New York, 1964.

\bibitem{horn2012matrix}
R.~A. Horn and C.~R. Johnson, \emph{Matrix Analysis}.\hskip 1em plus 0.5em
  minus 0.4em\relax Cambridge university press, 2012.

\bibitem{walter1998ordinary}
W.~Walter, \emph{Ordinary Differential Equations}, ser. Graduate Texts in
  Mathematics Series.\hskip 1em plus 0.5em minus 0.4em\relax Springer-Verlag,
  1998.

\bibitem{rudin1973functional}
W.~Rudin, \emph{Functional Analysis}.\hskip 1em plus 0.5em minus 0.4em\relax
  McGraw-Hill New York, 1973.

\bibitem{birkhoff1989ordinary}
G.~Birkhoff and G.~C. Rota, \emph{Ordinary Differential Equations}.\hskip 1em
  plus 0.5em minus 0.4em\relax John Wiley and Sons, 1989.

\bibitem{becerra2004solving}
V.~M. Becerra, ``Solving optimal control problems with state constraints using
  nonlinear programming and simulation tools,'' \emph{IEEE Transactions on
  Education}, vol.~47, no.~3, pp. 377--384, 2004.

\bibitem{kutz2005practical}
N.~J. Kutz, \emph{Practical Scientific Computing, AMATH 581 Course
  Notes}.\hskip 1em plus 0.5em minus 0.4em\relax Department of Applied
  Mathematics, University of Washington, 2005.

\end{thebibliography}

%
%
%
%
%
%
%

\end{document}